\documentclass[11pt, letterpaper]{article}
\usepackage[utf8]{inputenc}

\usepackage{graphicx}
\PassOptionsToPackage{numbers}{natbib}
\usepackage{natbib}

\usepackage{amsthm}
\usepackage{algpseudocode}
\usepackage{algorithm}
\usepackage[group-separator={,}]{siunitx}
\usepackage{doi}
\usepackage{xcolor}
\usepackage{authblk}
\usepackage{a4wide}
\usepackage{amssymb}
\usepackage{amsmath}
\usepackage{booktabs}
\usepackage[capitalize]{cleveref}
\usepackage{thm-restate}
\usepackage{cleveref}

\usepackage{tikz}
\usetikzlibrary{arrows.meta,calc,positioning}
\usepackage{xcolor}
\usepackage{subcaption} 

\newtheorem*{invariant*}{Invariant}
\newtheorem{lemma}{Lemma}
\newtheorem{claim}{Claim}
\newtheorem{corollary}{Corollary}
\newtheorem{observation}{Observation}

\crefname{claim}{Claim}{Claims}

\newcommand{\OPT}{\text{\textsc{Opt}}}
\newcommand{\E}{\mathcal{E}}
\newcommand{\ALG}{\text{\textsc{Alg}}}
\newcommand{\I}{\mathcal{I}}
\newcommand{\A}{\mathcal{A}}

\definecolor{MyBlue}{RGB}{31,119,180}
\definecolor{MyRed}{RGB}{214,39,40}
\definecolor{MyGray}{gray}{0.7}

\newcommand{\drawInterval}[5][]{%
  \draw[line width=1pt, draw=#1, #2] (#3,#5) -- (#4,#5);
  \fill[draw=#1, fill=#1] (#3,#5) circle (1.2pt);
  \fill[draw=#1, fill=#1] (#4,#5) circle (1.2pt); 
}

\newcommand{\machineRow}[3]{%
  \draw[MyGray, line width=0.6pt] (0,#1) -- (#3,#1);
  \node[left=4pt] at (0,#1) {#2};
}

\usepackage{placeins}
\usepackage{enumitem}

\title{Fair Division Meets Scheduling: Approximately Envy-Free Interval Scheduling}

\author[1]{Sander Borst\thanks{ \href{mailto:sborst@mpi-inf.mpg.de}{sborst@mpi-inf.mpg.de}}}
\author[2]{Golnoosh Shahkarami\thanks{ \href{mailto:gshahkar@mpi-inf.mpg.de}{gshahkar@mpi-inf.mpg.de}}}
\author[3]{Rohit Vaish\thanks{ \href{mailto:rvaish@iitd.ac.in}{rvaish@iitd.ac.in}}}

\affil[1]{Max Planck Institute for Informatics, Saarland Informatics Campus}
\affil[2]{Max Planck Institute for Informatics, University of Bremen}
\affil[3]{Indian Institute of Technology Delhi}

\date{}

\begin{document}

\maketitle
\begin{abstract}
We study interval scheduling from the perspective of fair allocation.
There are $m$ identical machines and a set of intervals, each specified by a start time, an end time, and a nonnegative weight.
A schedule assigns a subset of the intervals to the machines so that no two intervals on the same machine overlap, and the goal is to maximize the total weight of scheduled intervals.
Viewing machines as agents and intervals as goods, we require the schedule to be envy-free up to one item (EF1), and we measure efficiency against the offline optimum without fairness.

In the offline setting, we give an algorithm that computes an EF1 schedule whose loss is at most a factor of $3/2$ in the unweighted regime, and we prove lower bounds of $\frac{3m-2}{2m-1}$, approaching $3/2$, in both the unweighted and the unit-length weighted regimes, so the price of fairness is $3/2$ in the limit.
In the online setting, intervals arrive in nondecreasing order of start times; an arriving interval must be accepted or rejected, rejections are irrevocable, and an accepted interval may be revoked, and lost, at any time before it ends.
For the unweighted regime we present Greedy-Balanced, a simple algorithm that maintains EF1 at every point in time and is $(2-\tfrac{1}{m})$-competitive against the offline optimum without fairness, and we prove a matching lower bound for every deterministic algorithm; the optimal deterministic fair competitive ratio is thus exactly $2-\tfrac{1}{m}$.
Experiments on real-world benchmark instances show that Greedy-Balanced performs well beyond its worst-case guarantee, with an observed ratio never exceeding $1.306$.
\end{abstract}

\section{Introduction}

Interval scheduling is a fundamental optimization problem in which one must assign time intervals, each specified by a start time and an end time, to identical machines so that intervals assigned to the same machine do not overlap, while maximizing the total value of the accepted intervals~\citep{KolenLenstraPapadimitriouSpieksma2007,KOVALYOV2007331}. It arises naturally in two-sided service platforms. Consider a ride-hailing service such as Uber or Lyft, a food-delivery platform such as Deliveroo, or a freelance marketplace such as Upwork: requests for service arrive over time, each with a fixed window during which it must be handled, and the platform assigns each accepted request to one of several available workers. Since a worker cannot handle two overlapping requests, the set of requests served by any single worker must be conflict-free, which is precisely the feasibility constraint of interval scheduling, with workers playing the role of machines.
On such platforms, the classical objective of maximizing the number of served requests is not the only concern. The workers who fulfill these requests are stakeholders in their own right: since each accepted request represents one unit of paid work, workers reasonably expect comparable access to these opportunities. The concern persists when the compensation structure is reversed: in settings where workers receive a fixed salary, a worker cares not about obtaining more requests, but about not being assigned substantially more work than others for the same pay. Our setting covers both readings, treating requests either as goods that workers compete for or as chores that they would rather avoid. An efficiency-maximizing assignment may pile many requests onto a few workers while leaving others idle, even though all workers are equally capable. When machines represent workers, or more generally clients or organizational units sharing a common infrastructure, such concentration is naturally perceived as unfair. This tension is inherent to the two-sided structure: efficiency is measured over the requests we accept, whereas fairness is measured over the workers to whom we assign them. It motivates the central question:
\begin{quote}
\emph{Can one combine the efficiency of interval scheduling with the fairness guarantees studied in fair division, while respecting the underlying feasibility constraints?}
\end{quote}

The offline and online versions of interval scheduling have both been studied extensively. In the offline problem, all intervals are known in advance; the problem dates back to the 1950s, stemming from the works of Dantzig, Ford, and Fulkerson~\citep{DF54minimizing,FRO+24flows}, and has numerous applications including crew scheduling~\citep{BC98dynamic}, telecommunications~\citep{BNC+99bandwidth}, satellite photography~\citep{G95scheduling}, VLSI layout design~\citep{GLL79optimal}, and computational biology~\citep{CLR+05more}. In the online problem, intervals arrive sequentially in nondecreasing order of start times, and the algorithm must decide without knowledge of future arrivals which intervals to keep~\citep{Lipton1994OnlineIS}. We study a \emph{revocable} online model in which accepted intervals can be revoked at any time before their completion. This model captures settings with soft reservations or provisional allocations, where a task can be preempted or withdrawn before service is completed, but once the opportunity is declined it is lost. The performance of an online algorithm is measured by the \emph{competitive ratio}~\citep{ST85amortized}, comparing it to an offline algorithm with full knowledge of the input.

More broadly, interval scheduling can be understood as the allocation of time-indexed opportunities over shared resources, which naturally raises concerns about fairness alongside the classical goal of efficiency. Our work focuses on the \emph{unweighted} setting on identical machines, which is the simplest framework for integrating concepts from fair division and offers a relevant context for analyzing the online cost of fairness. In more complex weighted scenarios, the online interval scheduling problem does not permit a constant competitive ratio, making the unweighted setting ideal for isolating and quantifying the efficiency loss due to fairness. 

From this perspective, fair division guarantees can be imposed on temporal allocation problems as explicit constraints, with interval scheduling providing a clean first instantiation.

To reason about fairness, we draw on the area of \emph{fair division}, which provides a principled framework for allocating resources among agents~\citep{BT96fair,M04fair,BCE+16handbook}. A central notion in this literature is \emph{envy-freeness}~\citep{F67resource,V74equity}, which requires that no agent should prefer another agent's allocation over its own. For indivisible goods, exact envy-freeness may be impossible, so one turns to its relaxations. 
A well-studied notion of approximate fairness is 
\emph{envy-freeness up to one item} (EF1)~\citep{B11combinatorial}, which allows pairwise envy to be eliminated by removing a single item from the envied bundle. It is known that an EF1 allocation always exists for discrete goods~\citep{LMM+04approximately}. In our setting, we consider intervals as the goods and machines as the agents, and EF1 requires that any envy between two machines can be resolved by hypothetically removing one interval from the schedule with a higher value.

Our goal is not to replace the scheduling objective with a fairness objective; instead, we aim to combine both perspectives. Traditional approaches to fairness in optimization usually incorporate fairness directly into the objective, often using egalitarian or max-min criteria. In contrast, we impose a fair-division guarantee as a \emph{hard constraint} on the scheduling problem and compare the resulting efficiency to the unconstrained optimum. This approach provides a clean way to quantify the cost of fairness while preserving the combinatorial structure of interval scheduling.

In the offline setting, we measure this loss through the \emph{price of fairness}~\citep{BFT11price,CKK+12efficiency}, namely the worst-case ratio between the value achieved by a fair schedule and that of an optimal schedule without fairness constraints. In the online setting, the benchmark must additionally account for the uncertainty of future arrivals. For this reason, we formulate the \emph{fair competitive ratio}, defined as the worst-case ratio between the value achieved by a fair online algorithm and that of an optimal offline algorithm without fairness constraints. This benchmark captures both sources of degradation: the loss imposed by fairness and the additional loss due to online uncertainty.

With this framework in place, we study interval scheduling under EF1 in both offline and online settings, establish tight worst-case guarantees, and determine how much efficiency must be sacrificed to obtain fair allocations over time.

\subsection{Our Contributions and Techniques}

At a high level, our work can be viewed from two complementary perspectives. First, we can consider interval scheduling under EF1 as a \emph{fair allocation problem with scheduling constraints}. In this context, machines act as agents, intervals represent items, and the feasibility of a schedule is determined by the requirement that only non-conflicting items are assigned to each agent. Alternatively, we can see it as a \emph{scheduling problem with a fairness constraint}, where the goal is to maximize total accepted weight while adhering to EF1. The first perspective links our research to the fair division literature, while the second focuses on approximation guarantees and potential efficiency loss. 
A central theme of our paper is to reconcile these two viewpoints by allowing for partial allocation and providing a clear understanding of the efficiency trade-offs associated with achieving fairness.

\paragraph{Offline Version as a Constrained Fair Allocation Problem.}
We start by considering the offline version of our problem. In the absence of fairness constraints, the offline interval scheduling problem on identical machines can be solved optimally via polynomial-time algorithms~\citep{arkin_scheduling_1987, bouzina_interval_1996}. When fairness constraints are imposed, our problem can be viewed as a fair allocation problem with feasibility constraints~\citep{S21constraints}. Prior work on this topic, especially in the context of interval scheduling, has focused on \emph{existential} questions pertaining to EF1 allocations alongside an efficiency requirement such as \emph{maximality}~\citep{KEG+24fair,IMY25dividing,EGI+26fair}. (A maximal schedule is one where no unallocated interval can be feasibly assigned to any machine.) 
In contrast, we focus on optimizing a well-studied efficiency objective while allowing for partial allocations. Our goal is to maximize the total weight of accepted intervals subject to the EF1 constraint. 
This objective is strictly stronger than maximality, as any maximum (or efficient) allocation is also maximal. Our offline algorithms and lower bounds quantify the precise efficiency loss incurred by enforcing EF1 (i.e., the price of fairness) under scheduling constraints, yielding tight constant-factor bounds summarized in Table~\ref{tab:unit-weight}.

\paragraph{Online Arrival as an Additional Dimension.}
Next, we study how the fairness constraint interacts with online arrivals, introducing an additional layer of complexity beyond offline feasibility. 
In the online model, intervals arrive in nondecreasing order of their start times. Upon arrival, each interval must either be accepted or rejected. Once an interval is rejected, that decision is final; however, an accepted interval can be revoked at any time before it ends. A fair online algorithm must make decisions without knowledge of future inputs while adhering to fairness constraints. This setup allows us to isolate the efficiency loss caused by online uncertainty, over and above the inherent loss associated with enforcing EF1 fairness.

Our results in the online setting are closely related to the literature on \emph{temporal fair division}~\citep{ELL+25temporal,CES25temporal}. This line of work studies fairness guarantees that must be upheld at every time step, but typically assumes complete knowledge of inputs in advance. In contrast, we require EF1 to be maintained at all times while operating online, without access to future intervals. By establishing tight bounds for deterministic online algorithms, we clarify the extent of additional efficiency that must be sacrificed to uphold fairness when decisions are made irrevocably and under uncertainty.

\paragraph{Valuation Models and Results Overview.}
Previous research in online interval scheduling has shown that no constant-factor competitive guarantees are possible, even on a single machine and in the absence of fairness constraints, when intervals can have arbitrary weights~\citep{WOEGINGER19945}. Therefore, we focus on identical machines primarily in the \emph{unweighted} setting, where all intervals have the same weight. In the fair allocation literature, this is often referred to as \emph{uniform valuations} and is closely related to the equitable coloring problem~\citep{HS87proof}. 

Although uniform valuations may appear simple in the context of fair division, they represent a \emph{central and well-studied} model in scheduling, particularly in the online setting. In this scenario, maximizing the total accepted weight reduces to maximizing the total number of accepted intervals. This is a classic objective that remains challenging under online arrivals and feasibility constraints. Furthermore, under uniform valuations, the EF1 requirement coincides with the stronger notion of envy-freeness up to \emph{any} item (EFX), as all items have equal value. Within this setting, we obtain tight constant-factor guarantees in both the offline and online models; summarized in Table~\ref{tab:unit-weight}.

In addition to examining uniform valuations, we ask whether enforcing EF1 allows for a constant efficiency loss under more general valuations. We show that in the offline setting when weights take a constant number of distinct values, it is possible to achieve constant-factor EF1 approximations. Furthermore, we establish that the $3/2$ lower bound continues to hold even in the highly structured \emph{unit-length} setting with arbitrary weights. This shows that efficiency loss is due to the fairness constraint, not interval lengths or valuation complexity. 
In the deterministic online setting, an additional loss occurs due to sequential arrivals, which leads to a tight competitive ratio of $2-\tfrac{1}{m}$. Together, these results provide a clear separation between the efficiency loss caused by fairness and the loss associated with online uncertainty.

Although we state our results for goods with nonnegative weights, the unweighted results also transfer to the chore setting: under uniform valuations, EF1 again reduces to the same load-balancing condition, namely that any two machines receive numbers of intervals differing by at most one.

We also experimentally evaluate our online algorithm on real-world benchmark instances. These experiments show that the algorithm consistently outperforms its theoretical guarantees in practice. This suggests that while the worst-case competitive ratio is tight, the algorithm performs significantly better on typical instances.

\begin{table}[!h]
\centering

{\renewcommand{\arraystretch}{1.25}
\begin{tabular}{@{}lcc@{}} 
\toprule
& \multicolumn{1}{c}{Upper Bound} & \multicolumn{1}{c}{Lower Bound} \\ \midrule
Competitive Ratio (Online, No Fairness)    
& $1$ \cite{faigle_note_1995} 
& $1$ \\

Price of Fairness (Offline)             
& $\tfrac{3}{2}$ \scriptsize[T.~\ref{thm:offline-ef1-upper}] 
& $\tfrac{3}{2}$ \scriptsize[T.~\ref{thm:offline-ef1-lower}] \\

Fair Competitive Ratio (Online + Fair)              
& $2 - \tfrac{1}{m}$ \scriptsize[T.~\ref{thm:gb-ef1-competitiveness}] 
& $2-\tfrac{1}{m}$ \scriptsize[T.~\ref{thm:lb-deterministic-ef1}]  \\ 
\bottomrule

\end{tabular}
}
\caption{Summary of tight guarantees in the unweighted (uniform valuation) setting.
The optimal competitive ratio of $1$ for online interval scheduling without fairness is due to~\cite{faigle_note_1995}. All ratios are measured against the offline optimum \emph{without} fairness.}
\label{tab:unit-weight}
\end{table}

\paragraph{Technical Overview.}
With unit weights, EF1 has a clean characterization: a schedule is EF1 precisely when the numbers of intervals on any two machines differ by at most one (\Cref{obs:ef1-loads}).
So fairness never dictates which intervals a machine gets, only how many, and all of our arguments are really about one question: where does this count-balance requirement force the loss of intervals that an unconstrained scheduler would keep?

In the offline setting the answer is local.
We start from an efficiency-optimal schedule that ignores fairness, rebalance it by moving intervals to underloaded machines where no conflict interferes, and resolve the remaining conflicts by swaps that trade intervals across machines; each swap sacrifices at most one interval, which yields the upper bound, and a matching lower bound shows these sacrifices are unavoidable.

In the online setting our algorithm, Greedy-Balanced, separates efficiency from fairness as follows.
It runs in the background a classical revocation-based greedy that is optimal for unweighted interval scheduling without fairness, and uses it as a filter: an arriving interval is considered only if the benchmark accepts it.
Placement then follows a simple principle: if the interval fits on a least-loaded machine without violating EF1, allocate it there; otherwise, repair the schedule by replacing, on a least-loaded machine, a conflicting interval with a later right endpoint, which makes the light machines more available for future arrivals by improving the endpoints of what they hold; if neither is possible, reject.
The analysis of this algorithm is where the main technical machinery of the paper lives, and it rests on three devices: a decomposition of the execution into \emph{time blocks}, an \emph{auxiliary graph} that records where each loss happened, and a \emph{charging argument} on this graph that caps the losses per block.
We describe each in turn.

\emph{Time blocks.}
We cut the execution at the moments $t_k$ when the minimum machine load first reaches $k$, and call the period between $t_k$ and $t_{k+1}$ a block.
Blocks are the natural unit of progress: within one block the fairness constraint has a frozen shape, machines at the minimum load are the only ones allowed to grow, and by the end of the block every machine has grown by exactly one, so the algorithm collects exactly $m$ intervals per block.
The whole analysis then reduces to a single question about a single block: how many intervals that the benchmark keeps can the algorithm be forced to give up before the minimum load rises?

\emph{The auxiliary graph.}
To answer it we build, for each block, a graph that records where each loss happened.
Its vertices are the machines, and it has one edge per \emph{loss event} of the block, the moments at which the algorithm discards an interval that survives in the benchmark's final schedule.
The edge of a loss connects two machines in distinct roles: its \emph{receiver}, the light machine on which the algorithm placed the arriving interval (a dummy node when the interval was rejected outright), and its \emph{provider}, a machine that could have accepted the arriving interval had fairness not been in the way, as it has no conflict with it but is already heavy.
Every loss of the block is now captured by an edge, so bounding the losses becomes a purely combinatorial question: how many edges can this graph contain?

\emph{The charging argument.}
The answer comes from a charging rule of pleasant simplicity: each receiver is charged the last incoming edge it receives, and every other edge is charged to its provider.
Two exchange properties of the graph make this work: when a provider creates a new outgoing edge, its previous receiver must already have turned heavy and can receive no further edge, and when a receiver acquires a second incoming edge, its previous provider can create no further edge.
Together they imply that every edge charged to a provider is the last edge leaving it, so every edge is the last edge of the vertex that pays for it, and each machine pays at most once.
Moreover, some machine pays nothing: if every machine were charged, every machine would already be heavy before the last loss of the block, so the block would have ended earlier.
Hence a block has at most $m-1$ losses against exactly $m$ gains, and summing over blocks gives the fair competitive ratio $2-\tfrac{1}{m}$; the per-block bound is attained, so the analysis of the algorithm is tight.
\Cref{fig:overview-charging} illustrates the rule.

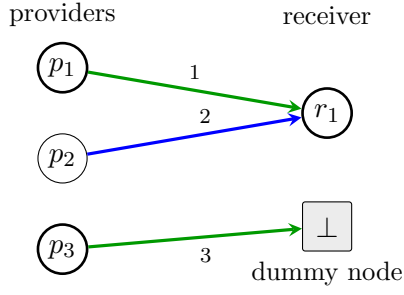
\begin{figure}[t]
\centering
\begin{tikzpicture}[
    x=1cm,y=1cm,
    provider/.style={circle, draw, minimum size=6.5mm, inner sep=0pt},
    chargedprovider/.style={circle, draw, minimum size=6.5mm, inner sep=0pt, line width=1pt, font=\bfseries},
    chargedreceiver/.style={circle, draw, minimum size=6.5mm, inner sep=0pt, line width=1pt, font=\bfseries},
    dummy/.style={rectangle, draw, rounded corners=1pt, minimum width=6.5mm, minimum height=6.5mm, inner sep=1pt, fill=gray!15},
    >=stealth
]
\node[chargedprovider] (p1) at (0,2.1)  {$p_1$};
\node[provider]        (p2) at (0,0.9)  {$p_2$};
\node[chargedprovider] (p3) at (0,-0.3) {$p_3$};
\node[chargedreceiver] (r1) at (3.5,1.5) {$r_1$};
\node[dummy]           (d)  at (3.5,0.0) {$\bot$};
\draw[->, very thick, green!60!black] (p1) -- node[above, font=\scriptsize, black] {1} (r1);
\draw[->, very thick, blue] (p2) -- node[above, font=\scriptsize, black, pos=0.55] {2} (r1);
\draw[->, very thick, green!60!black] (p3) -- node[below, font=\scriptsize, black, pos=0.55] {3} (d);
\node[font=\small] at (0,2.8) {providers};
\node[font=\small] at (3.5,2.8) {receiver};
\node[font=\small] at (3.5,-0.62) {dummy node};
\end{tikzpicture}
\caption{The charging rule for $m=4$; edges are numbered by creation time.
Edge~2, the last incoming edge of $r_1$, is charged to $r_1$ (blue); edges~1 and~3 are charged to their providers (green), and each is the last edge leaving its provider.
Every machine pays at most once, and $p_2$, whose only edge is claimed by its receiver, pays nothing, so the block has at most $m-1=3$ losses.}
\label{fig:overview-charging}
\end{figure}

\emph{The lower bound: one machine that starves itself.}
The matching lower bound shows that no deterministic online EF1 algorithm beats $2-\tfrac1m$, so Greedy-Balanced is optimal.
The adversary also works in blocks, though of a different kind: disjoint time windows fixed in advance rather than periods defined by the algorithm's behavior.
In each block it watches one fixed machine, say machine~$1$, and repeatedly releases the first half of a shrinking window: if machine~$1$ takes the interval, the next one is released nested inside it, so accepting it forces machine~$1$ to revoke what it just took; if machine~$1$ passes, the window moves past the interval, which then conflicts with nothing that comes later.
A block ends once machine~$1$ has accepted $m$ times, keeping only the innermost interval, while the optimum keeps everything: the nested chain fits across the $m$ machines, and each passed interval ends exactly where the next one begins.
Refusing to engage does not help, since EF1 keeps every other machine within one interval of machine~$1$, so further releases are pure gifts to the optimum.
Each block therefore hands the optimum $m-1$ intervals more than the algorithm, and over many blocks the ratio tends to $2-\tfrac1m$.

\subsection{Related Work}
The online interval scheduling problem was introduced by~\citet{Lipton1994OnlineIS}. In their model, the intervals arrive in the order of their start times, and must be accepted or rejected irrevocably upon their arrival, following a non-overlap constraint. Early work on this problem showed that no deterministic algorithm can achieve a constant competitive ratio for irrevocable scheduling on a single machine~\citep{MSS-Long-1993}. Similarly, for randomized algorithms, a lower bound of $\mathcal{O}(\log \Delta)$ was shown, where $\Delta$ is the ratio of the length of the longest and shortest intervals~\citep{Lipton1994OnlineIS}.
Subsequent studies explored the idea of \emph{revoking} previously accepted intervals, or \emph{preemption}, to achieve improved competitive guarantees~\citep{faigle_note_1995,GARAY1997180,TOMKINS1995173,10.1007/978-3-031-49815-2_13}. In particular, \citet{faigle_note_1995} studied revocable online interval scheduling for multiple identical machines. They showed that a simple greedy ``deadline replacement'' strategy achieves a competitive ratio of $1$ in the unweighted case.

In the \emph{weighted} setting, arbitrary weights preclude finite competitive ratios in general, both for deterministic and randomized algorithms \citep{WOEGINGER19945,CanettiIrani1998}. Positive results have been shown under structured settings, such as when the weight of each interval is correlated with its length~\citep{WOEGINGER19945}. 

\citet{fung_-line_2012} study online scheduling of \emph{unit-length} intervals with \emph{arbitrary weights} on multiple identical machines, giving a greedy algorithm with competitive ratio approaching $2$, which matches the lower bound without fairness constraints~\citep{fung2008online}. We refer the reader to the surveys by \citet{KolenLenstraPapadimitriouSpieksma2007} and \citet{KOVALYOV2007331} for a comprehensive overview of interval scheduling and related allocation problems. 
Recent research has focused on online interval scheduling on a single machine within a \emph{learning-augmented} framework, also referred to as algorithms with predictions~\citep{10.1007/978-3-031-38906-1_14, karavasilis2025intervalselectionbinarypredictions, antoniadis2025switchingframeworkonlineinterval}. This area of study aims to develop online algorithms whose competitive performance improves when the predictions are accurate and degrades gracefully when the predictions are less reliable.

We now turn our attention to the literature on fair division. In the \emph{offline} setting, recent works have studied the fair allocation of discrete items under conflict constraints~\citep{LLZ21fair,HH22fair,KEG+24fair,IMY25dividing,EGI+26fair}. In this context, all items and their pairwise conflicts are provided to the algorithm in advance. Due to the offline nature of the problem, the focus shifts from competitive analysis to investigating the \emph{existence} of a desirable, feasible allocation; specifically, one that is approximately envy-free and economically efficient~(e.g., maximal or Pareto optimal). 

The \emph{online} fair division problem has been extensively studied in recent years. Several works have explored achieving either exact or approximate envy-freeness alongside Pareto optimality as indivisible items are presented over time~\citep{AAG+15online, AW19strategy, HPP+19achieving,AW20online, BKP+24fair,HHI+24class,KMN+26online}. With the exception of \cite{HPP+19achieving}, these studies require that the items be assigned irrevocably. Additionally, some research focuses on the \emph{repeated} fair division problem, where the same set of items is presented to the same set of agents at each time step~\citep{ILN+24repeated, CN24repeatedly}. 
The studies by~\cite{ELL+25temporal} and \cite{CES25temporal} focus on the problem of \emph{temporal fair division}. In this context, the algorithm is aware of the order in which items will arrive and must ensure that the allocation remains ``anytime fair'', i.e., satisfies approximate envy-freeness cumulatively at every time step. \cite{CL26temporal} study a generalization of this model where an item arriving at time $t$ can be assigned anytime in the next $r$ steps. Notably, these works do not take conflict constraints into account.

\section{Preliminaries}
\label{sec:prelim}
\paragraph{Interval Scheduling Model.}
We consider interval scheduling on $m$ identical machines. The input consists of a set $\I$ of $n$ intervals, where each interval $I_j \in \I$ is specified by a start time $s_j$, an end time $e_j$, and a nonnegative weight $w(I_j)$. We identify each interval with the half-open set $[s_j,e_j)$, and we say that two intervals \emph{overlap} (or \emph{conflict}) if they intersect. We assume all start and end times are nonnegative rationals in the range $[0,T]$ for some $T\in\mathbf{Q}_{\geq 0}$. The length of an interval is $l_j = e_j - s_j$. A schedule assigns accepted intervals to machines such that no two intervals assigned to the same machine overlap; intervals assigned to different machines may overlap.
\paragraph{Offline and Online Models.}
In the \emph{offline} setting, the entire set $\I$ is known in advance. In the \emph{online} setting, intervals arrive sequentially in nondecreasing order of their start times. When an interval $I_j$ arrives at time $s_j$, the algorithm learns $(s_j,e_j)$ and $w(I_j)$, and must immediately decide whether to accept or reject it, without knowledge of future intervals. Once an interval is rejected, this decision is irrevocable. We allow \emph{revocation} in the online model: an accepted interval may be removed from the schedule at any time before its end time $e_j$. A revoked interval cannot be scheduled again; from that point on, it is treated as rejected.
\paragraph{Machines and Assignments.}
Scheduling takes place on $m$ identical machines, which we denote by $\A=\{1,2,\dots,m\}$. If an interval is accepted, it must be assigned to a machine $i\in\A$. We assume without loss of generality that intervals are indexed in order of arrival, so that $s_1 \le s_2 \le \dots \le s_n$. For a given schedule, we denote by $S_i$ the set of intervals assigned to machine $i$, and we refer to $|S_i|$ as the load of machine $i$.
\paragraph{Efficiency Objective.}
For a set $X$ of intervals, we write $w(X)=\sum_{I_j\in X} w(I_j)$ for its total weight. For an algorithm $\ALG$ and an instance $\I$, we denote by $\ALG(\I)$ the set of intervals in the final schedule of $\ALG$, that is, the intervals that $\ALG$ accepts and never revokes, and we write $|\ALG(\I)|$ for the total weight $w(\ALG(\I))$. The goal is to maximize the total weight of the final schedule. In the \emph{unweighted} setting, all intervals have unit weight, and $|\ALG(\I)|$ is simply the number of scheduled intervals. This setting is also referred to as uniform valuation. We denote by $\OPT$ the optimal offline algorithm without fairness constraints, so that $\OPT(\I)$ is an optimal schedule of instance $\I$ and $|\OPT(\I)|$ is its total weight.
\paragraph{Performance Measures.}
We now define the measures used to evaluate the performance of our algorithms in the offline and online settings.
\begin{itemize}
     \item{\emph{Competitive Ratio}}: In the \emph{online} setting, the algorithm is provided the intervals sequentially. The performance of a (possibly unfair) online algorithm $\ALG^\texttt{online}{}$ is measured by the \emph{competitive ratio}, defined as the smallest $\alpha$ such that for every input instance $\I$ we have $|\OPT(\I)| \le \alpha \cdot |\ALG^\texttt{online}{}(\I)|$. An algorithm is \emph{$c$-competitive} if its competitive ratio is at most $c$.
     \item{\emph{Price of Fairness}}: In the \emph{offline} setting, the algorithm is provided all intervals upfront. The performance of an offline fair algorithm $\ALG{}^\texttt{offline-fair}$ is measured by its price of fairness, defined as the smallest $\alpha$ such that for every input instance $\I$ we have $|\OPT(\I)| \le \alpha \cdot |\ALG{}^\texttt{offline-fair}(\I)|$.
     \item{\emph{Fair Competitive Ratio}}: This performance measure evaluates the efficiency loss in the presence of both \emph{fairness} and \emph{online} constraints. Let $\ALG^\texttt{online-fair}{}$ denote an online algorithm that satisfies EF1. Then, the \emph{fair competitive ratio} of such an algorithm is defined as the smallest $\alpha$ such that for every input instance $\I$ we have $|\OPT(\I)| \le \alpha \cdot |\ALG^\texttt{online-fair}{}(\I)|$.
\end{itemize}
\paragraph{Benchmark.}
Throughout the paper, all performance guarantees are stated with respect to the offline optimum $|\OPT(\I)|$ computed without fairness constraints. Let $\E$ denote the greedy online algorithm of \citet{faigle_note_1995}. In the unweighted setting, $\E$ attains the offline optimum, i.e., $|\E(\I)| = |\OPT(\I)|$.
\paragraph{Fairness Notions.}
In addition to efficiency, we study fairness across machines. We view each machine as an agent and the intervals assigned to it as the agent's bundle. The utility that machine $i$ derives from a set $S$ of intervals is additive, i.e., $u_i(S) = \sum_{I_j \in S} w(I_j)$. We adopt envy-based fairness notions from the fair division literature. A schedule is \emph{envy-free up to one interval (EF1)} if for every pair of machines $i,k\in\A$ with $i\neq k$ and $S_k \neq \emptyset$, there exists an interval $I\in S_k$ such that $u_i(S_i) \ge u_i(S_k) - w(I)$. A schedule is \emph{envy-free up to any interval (EFX)} if for every pair $i,k\in\A$ with $S_k \neq \emptyset$ and every interval $I\in S_k$, $u_i(S_i) \ge u_i(S_k) - w(I)$. In the unweighted setting, EF1 and EFX coincide. These notions are inspired by the \emph{envy-freeness up to one good}~\citep{LMM+04approximately,B11combinatorial} and \emph{envy-freeness up to any good}~\citep{CKM+19unreasonable} notions in the fair division literature.
\begin{observation}\label{obs:ef1-loads}
In the unweighted setting, a schedule is EF1 if and only if the loads of any two machines differ by at most one.
\end{observation}

\section{Offline EF1 Scheduling}
\label{sec:offline}
We begin by studying the offline version of interval scheduling under EF1 fairness, where the entire set of intervals is known in advance. As discussed earlier, this problem can be viewed both as a fair allocation problem with scheduling constraints and as a scheduling problem with a fairness constraint. A central question in this setting is how much efficiency must be sacrificed in order to guarantee EF1. Our goal in this section is to characterize the loss of efficiency caused by enforcing fairness in the offline setting, by designing approximation algorithms for EF1 scheduling and by proving matching lower bounds that establish the intrinsic cost of EF1.
\subsection{Maximum EF1 Scheduling}
Without fairness constraints, the offline interval scheduling problem on identical machines admits an efficient algorithm computing an optimal schedule~\citep{arkin_scheduling_1987, bouzina_interval_1996,CARLISLE1995225}. This provides a natural benchmark for evaluating the cost of enforcing fairness. We first consider the \emph{unweighted} setting, where all intervals have unit weight. Our focus is on computing a schedule that satisfies EF1 while retaining as much efficiency as possible. In particular, we show that EF1 can be enforced in the offline setting with only a $3/2$ price of fairness.
\begin{restatable}[Offline EF1 Upper Bound: Unweighted]{theorem}{OfflineUpperBoundUnitWeight}
In the \emph{unweighted offline} setting with $m$ identical machines, there exists a polynomial-time EF1 scheduling algorithm $\ALG$ whose price of fairness is at most $3/2$. That is, for every instance $\I$,
\[
\frac{|\OPT(\I)|}{|\ALG(\I)|} \;\le\; \frac{3}{2},
\]
where $\OPT(\I)$ denotes an optimal offline schedule without fairness constraints.
\label{thm:offline-ef1-upper}
\end{restatable}
We next describe the algorithm and prove Theorem~\ref{thm:offline-ef1-upper}.
\paragraph{\emph{Rebalance-and-Swap} Algorithm.}
Our algorithm starts from an optimal offline schedule without fairness constraints and incrementally transforms it into an EF1 schedule. This makes the efficiency loss due to fairness easier to analyze, as we can directly compare to the initial optimal schedule.
This is accomplished via a subroutine \Call{Rebalance}{}, that for a given set of machines $\A$, produces a new schedule such that the load on each machine in $\A$ is either $\lfloor\mu\rfloor$ or $\lceil\mu\rceil$, where $\mu$ is the average load across machines in $\A$ after removing $(m-1)$ intervals. It does this by repeatedly considering the two machines with the largest and smallest load, and invoking a procedure \Call{Swap}{} that exchanges intervals between these two machines such that the load on one of the machines becomes either $\lfloor\mu\rfloor$ or $\lceil\mu\rceil$. The machine with the balanced load is then removed from consideration, and the process continues until all machines in $\A$ have load either $\lfloor\mu\rfloor$ or $\lceil\mu\rceil$.
The key component in our algorithm is the \Call{Swap}{} procedure that takes as input two machines $i$ and $\bar{i}$, along with target loads $t_i$ and $t_{\bar{i}}$. It works by taking a time $t$ and swapping intervals that start after time $t$ between the two machines. We choose the time $t$ such that we only need to discard one interval from one of the machines to achieve the target loads. The full description of these procedures are given below.
\begin{algorithm}
\caption{\emph{Rebalance-and-Swap} Algorithm}
\begin{algorithmic}[1]
\State \textbf{Input:} Interval set $\I$ and $m$ identical machines.
\State \textbf{Result:} An EF1 schedule of intervals on machines.
\Procedure{FairOffline}{$\I, m$}
\State Compute an optimal offline schedule $\OPT$ without fairness constraints.
\State Initialize $S_i$ to be the set of intervals assigned to machine $i$ in $\OPT$.
\While{$\exists$ machines $i, \bar{i}$ with $|S_i|\geq 2$ and $|S_{\bar{i}}|=0$}
    \State Reassign an interval from $i$ to $\bar{i}$.
\EndWhile
\State Reorder the machines so that $|S_1|\geq |S_2|\geq \ldots \geq |S_m|$.
\If{$|\OPT|\leq \frac32 m$}
    \State For each machine $i$ with $|S_i|\geq 2$, drop intervals until $|S_i|=1$.
    \State\Return current schedule.
\ElsIf{$|\OPT|\leq 3m -1$}
    \While{there exist machines $i, \bar{i}$ with $|S_i|\geq 4$ and $|S_{\bar{i}}|=1$}
    \State Choose such machines $i$ and $\bar{i}$.
    \State \Call{Swap}{$i, \bar{i}, |S_i|-2, 2$}
    \EndWhile
    \State For $i=1,\ldots, m$ drop intervals from $S_i$ until $|S_i|\leq 2$.
\Else
\State  \Call{Rebalance}{$\{1,\ldots, m\}$}
\EndIf
\EndProcedure
\end{algorithmic}
\end{algorithm}
\begin{algorithm}
\caption{Procedure Rebalance}
\begin{algorithmic}[1]
\State \textbf{Input:} Machine set $\A$.
\State \textbf{Result:} New schedule $S'$ satisfies $|S'_i|\in\{\lfloor\mu\rfloor ,\lceil\mu\rceil\}$ and $\sum_{i\in \A}|S'_i|=|\A|\cdot \mu$, where $\mu=\frac{\sum_{i\in \A} |S_i| - (|\A|-1)}{|\A|}$.
\Procedure{Rebalance}{$\A$}
    \State Set $\mu:= \frac{\sum_{i\in \A} |S_i| - (|\A|-1)}{|\A|}$.
    \State Choose $k$ such that $k\lfloor\mu\rfloor +(|\A|-k) \lceil\mu\rceil =| \A|\mu $.
    \State $F:= \emptyset$.
    \For {$r=1$ to $|\A|-1$}
    \State Let $i:= \arg\max_{i\in \A\setminus F} |S_i|$ and $\bar{i}:= \arg\min_{i\in \A\setminus F} |S_i|$.
    \State $t:=\lfloor \mu\rfloor$ if $r\leq k$ and $t:=\lceil\mu\rceil$ otherwise.
    \State $t':=|S_i|+|S_{\bar{i}}|-1-t$.
    \State \Call{Swap}{$i, \bar{i}, t, t'$}
    \State $F:= F\cup \{i\}$.
    \EndFor
\EndProcedure
\end{algorithmic}
\end{algorithm}
\begin{algorithm}
\caption{Procedure Swap}
\begin{algorithmic}[1]
\State \textbf{Input:} Machines $i, \bar{i}$ and target loads, $t_i, t_{\bar{i}}\in [\min(|S_{i}|, |S_{\bar{i}}|)-1, \max( |S_{i}|,|S_{\bar{i}}|)]$
\State \hspace{1.6cm} with $t_i+t_{\bar{i}}=|S_i|+|S_{\bar{i}}|-1$.
\State \textbf{Result:} $|S_{i}|=t_{i}$ and $|S_{\bar{i}}|=t_{\bar{i}}$.
\Procedure{Swap}{$i, \bar{i}, t_{i}, t_{\bar{i}}$}
\If{$t_i<t_{\bar{i}}$}
    \State \Call{Swap}{$\bar{i}, i, t_{\bar{i}}, t_i$}
    \State \Return
\EndIf
\If{$|S_i|<|S_{\bar{i}}|$}
    \State Swap all interval assignments of machines $i$ and $\bar{i}$.
\EndIf
\State Let $A_i(t):=\{j\in S_i:e_j \leq t\}$ and $B_i(t):=\{j\in S_i:s_j \geq t\}$ for both machines $i$ and $\bar{i}$. \label{line:swap-define-sets}
\State Find the smallest interval endpoint $t$ such that $|A_{i}(t)|-|A_{\bar{i}}(t)|=|S_{i}|-t_{i}=t_{\bar{i}}-|S_{\bar{i}}|+1$. \label{line:swap-find-t}
\State Set $S_{i}':=A_{\bar{i}}(t) \cup B_{i}(t)$ and $S_{\bar{i}}':=A_{i}(t) \cup B_{\bar{i}}(t)$.
\State If $|S_{\bar{i}}'|>t_{\bar{i}}$, remove one interval from $S_{\bar{i}}'$.
\State Assign $S_{i}:=S_{i}'$ and $S_{\bar{i}}:=S_{\bar{i}}'$.
\EndProcedure
\end{algorithmic}
\end{algorithm}
First we verify that the key subroutine of the algorithm, \Call{Swap}{}, is well-defined and achieves the desired effect.
\begin{lemma}
    For machines $i, \bar{i}$ and target loads, $t_i, t_{\bar{i}}\in [\min(|S_{i}|, |S_{\bar{i}}|)-1, \max( |S_{i}|,|S_{\bar{i}}|)]$ with $t_i+t_{\bar{i}}=|S_i|+|S_{\bar{i}}|-1$, \Call{Swap}{$i, \bar{i}, t_i, t_{\bar{i}}$} is well-defined and changes the schedule such that $|S_{i}|=t_{i}$ and $|S_{\bar{i}}|= t_{\bar{i}}$.\label{lem:swap-well-defined}
\end{lemma}
\begin{proof}
    By the first two if statements, we may assume without loss of generality that $|S_i|\geq t_i \geq t_{\bar{i}}\geq |S_{\bar{i}}|-1$ from line \ref{line:swap-define-sets} onwards.
    To see that there exists a feasible time $t$ in line~\ref{line:swap-find-t} and that the procedure is well-defined, observe that the function $|A_{i}(t)|-|A_{\bar{i}}(t)|$ takes values in $\mathbf{Z}$ and changes only at end times of intervals, so it suffices to consider the finitely many interval endpoints. Moreover, since intervals assigned to the same machine are disjoint, at most one interval of each machine ends at any given time; hence the function has jumps of size at most $1$. Furthermore, we have $|A_{i}(0)|-|A_{\bar{i}}(0)|=0$ and $|A_{i}(T)|-|A_{\bar{i}}(T)|=|S_i|-|S_{\bar{i}}|\geq 0$. Since $|S_i|-|S_{\bar{i}}| \geq |S_i|-t_{i} \geq 0$, there must be some time $t$ where $|A_{i}(t)|-|A_{\bar{i}}(t)|=|S_i|-t_{i}$.
    If $|S_i|-t_i=0$, then we may take $t=0$, and since all start times are nonnegative we have $A_i(t)\cup B_i(t)=S_i$.
    Otherwise, as $t$ is the smallest such time, an interval finishes on machine $i$ at time $t$. Hence, we again have $A_{i}(t)\cup B_{i}(t)=S_{i}$.
    For machine $\bar{i}$, we have $|S_{\bar{i}}\setminus (A_{\bar{i}}(t)\cup B_{\bar{i}}(t)) | \leq 1$, since at most one interval can be active on machine $\bar{i}$ at time $t$.
    So, we have $|S_{i}'|=|A_{\bar{i}}(t)|+|B_{i}(t)|=|A_{\bar{i}}(t)|+|S_{i}|-|A_{i}(t)|=t_{i}$, and $|S_{\bar{i}}'|\geq |A_{i}(t)|+|B_{\bar{i}}(t)|\geq |A_{i}(t)|+|S_{\bar{i}}|-|A_{\bar{i}}(t)|-1=t_{\bar{i}}$. Then if $|S_{\bar{i}}'|>t_{\bar{i}}$, we can remove an interval from $S_{\bar{i}}'$ so that $|S_{\bar{i}}'|=t_{\bar{i}}$. So, after putting $S_i:=S_{i}'$ and $S_{\bar{i}}:=S_{\bar{i}}'$, we have $|S_{i}|=t_{i}$ and $|S_{\bar{i}}|=t_{\bar{i}}$, as desired.
\end{proof}
Now we verify that the main procedure \Call{Rebalance}{} is well-defined, results in a schedule with the desired properties, and that all calls to \Call{Swap}{} within it satisfy the preconditions of \cref{lem:swap-well-defined}.
\begin{lemma}
The procedure \Call{Rebalance}{$\A$} is well-defined and results in a schedule $S'$ with $|S'_i|\in \{\lfloor\mu\rfloor ,\lceil\mu\rceil\}$ for all $i\in \A$ and $\sum_{i\in \A}|S'_i|=|\A| \cdot \mu$, where $\mu=\frac{\sum_{i\in \A} |S_i| - (|\A|-1)}{|\A|}$.
\label{lem:rebalance-well-defined}
\end{lemma}
\begin{proof}
  We first verify that the procedure \Call{Rebalance}{} is well-defined by checking that each call to \Call{Swap}{} satisfies its preconditions.
  We first check that for each call \Call{Swap}{$i, \bar{i}, t_i, t_{\bar{i}}$}, both target loads lie in the interval
  \[
      [\min(|S_i|,|S_{\bar{i}}|)-1,\max(|S_i|,|S_{\bar{i}}|)]
  \]
  and satisfy $t_i+t_{\bar{i}}=|S_i|+|S_{\bar{i}}|-1$.
  Since $i$ and $\bar{i}$ are chosen as the maximum-load and minimum-load machines in $\A\setminus F$, respectively, we have $|S_i|\ge |S_{\bar{i}}|$, so this interval is simply $[|S_{\bar{i}}|-1,|S_i|]$.
  From the definition of $t$ and $t'$ the latter follows immediately. Our algorithm maintains that at the start of the $r$th iteration we have $|\A\setminus F|=|\A|-r+1$ and $\sum_{i\in \A\setminus F} |S_i| = \lfloor \mu \rfloor \max(0, k-r+1)+\lceil \mu \rceil \min(|\A|-k, |\A| -r+1) + (|\A|-r)$. Therefore, in the $r$th iteration we have:
  \begin{align*}
    \min_{i\in \A\setminus F} |S_i| &\leq \left\lfloor \frac{\sum_{i\in \A\setminus F} |S_i|}{|\A\setminus F|} \right\rfloor \leq  \left\lfloor \frac{\lceil \mu\rceil (|\A|-r+1) + (|\A|-r)}{|\A|-r+1}   \right\rfloor =
     \lceil\mu \rceil .
  \end{align*}
  and
  \begin{align*}
    \max_{i\in \A\setminus F} |S_i| &\geq \left\lceil \frac{\sum_{i\in \A\setminus F} |S_i|}{|\A\setminus F|}\right\rceil \geq  \left\lceil \frac{\mu(|\A|-r+1)}{|\A|-r+1}\right\rceil=
     \lceil\mu \rceil.
  \end{align*}
  Hence, $|S_{\bar{i}}|-1 \leq \lfloor \mu \rfloor\leq t\leq \lceil \mu \rceil \leq |S_{i}|$. Furthermore, since $t'=|S_i|+|S_{\bar{i}}|-1-t$, the upper bound $t\le |S_i|$ implies
  \[
      t'\ge |S_{\bar{i}}|-1,
  \]
  and the lower bound $t\ge |S_{\bar{i}}|-1$ implies
  \[
      t'\le |S_i|.
  \]
  Thus both $t$ and $t'$ lie in $[|S_{\bar{i}}|-1, |S_i|]$. The target pair passed to \Call{Swap}{} is $(t,t')$, and the two targets sum to $|S_i|+|S_{\bar{i}}|-1$. We conclude that the preconditions of \Call{Swap}{} are satisfied, and the procedure is well-defined.
  It is also clear that for all $i\in F$, we have $|S_i|\in \{\lfloor\mu\rfloor ,\lceil\mu\rceil\}$. After the loop, there will be just one machine $i\in \A\setminus F$. This machine will have load $\lceil \mu \rceil$.
  So, we end up with $k$ machines with load $\lfloor \mu \rfloor$ and $|\A|-k$ machines with load $\lceil \mu \rceil$, and hence $\sum_{i\in \A}|S'_i|=|\A|\cdot \mu$, as desired.
\end{proof}
Now we can prove Theorem~\ref{thm:offline-ef1-upper}, showing that the above algorithm achieves EF1 with price of fairness at most $3/2$.
\begin{proof}[Proof of Theorem~\ref{thm:offline-ef1-upper}]
It is straightforward to verify that the algorithm returns an EF1 schedule and runs in polynomial time.
We analyze the schedule returned by the \emph{Rebalance-and-Swap} algorithm according to the value of $|\OPT(\I)|$.
Recall from \Cref{obs:ef1-loads} that in the unweighted setting, a schedule is EF1 if and only if the loads of any two machines differ by at most one.
We show that in all cases the algorithm outputs an EF1 schedule and satisfies
\[
|\ALG(\I)| \;\ge\; \tfrac{2}{3}|\OPT(\I)|.
\]
\paragraph{Case 1: $|\OPT(\I)| \le \tfrac{3}{2}m$.}
In this regime, the algorithm assigns at most one interval to each machine.
If $|\OPT(\I)|\le m$, then each interval can be placed on a distinct machine and the algorithm returns the same schedule, so $|\ALG(\I)|=|\OPT(\I)|$.
If $m<|\OPT(\I)|\le \tfrac{3}{2}m$, then the algorithm assigns exactly one interval to each machine, yielding $|\ALG(\I)|=m$.
Since $m\ge \tfrac{2}{3}|\OPT(\I)|$ in this range, the approximation guarantee holds.
In both subcases, all machines have equal load in $\ALG(\I)$, so the schedule is EF1.
\paragraph{Case 2: $\tfrac{3}{2}m < |\OPT(\I)| \le 3m-1$.}
After the initial reassignment step, every machine is nonempty; otherwise, since $|\OPT(\I)|>m$, there would still be an empty machine and another machine with at least two intervals.
The algorithm therefore maintains that every machine has at least one interval.
Each call to \Call{Swap}{} in this case is applied to a machine with load $L\ge 4$ and a machine with load~$1$, with target loads $L-2$ and~$2$.
Thus the call loses exactly one interval, changes the singleton machine into a machine with two intervals, and this machine is not modified again by the while loop or by the final dropping step.
When the while loop terminates, there are two cases.
First, suppose there is no machine with exactly one interval. Then every machine has load at least~$2$, and after the final dropping step every machine has exactly two intervals. Hence $|\ALG(\I)|=2m$, and since $|\OPT(\I)|\le 3m-1$, we have
\[
|\ALG(\I)|=2m \ge \tfrac{2}{3}|\OPT(\I)|.
\]
Second, suppose at least one singleton machine remains. Since the while loop has terminated, no machine has load at least~$4$. Thus before the final dropping step every machine has load in $\{1,2,3\}$.
The only losses in the final dropping step come from machines of load~$3$, each of which loses one interval and ends with load~$2$.
Charge each swap loss to the singleton machine that the swap changes to load~$2$, and charge each final dropped interval to the machine from which it is dropped.
These charges are all to distinct machines that have load~$2$ in the final schedule. Therefore
\[
|\OPT(\I)|-|\ALG(\I)|
\le
\#\{i: |S_i|=2 \text{ in the final schedule}\}
\le
\tfrac12 |\ALG(\I)|.
\]
Consequently, $|\ALG(\I)|\ge \tfrac{2}{3}|\OPT(\I)|$.
In both cases, every machine receives either one or two intervals in the final schedule, so the schedule is EF1.
\paragraph{Case 3: $|\OPT(\I)| \ge 3m$.}
Here the algorithm produces a schedule in which machines receive either
\[
L_- \;=\; \Bigl\lfloor \tfrac{|\OPT(\I)|-(m-1)}{m} \Bigr\rfloor
\quad\text{or}\quad
L_+ \;=\; \Bigl\lceil \tfrac{|\OPT(\I)|-(m-1)}{m} \Bigr\rceil
\]
intervals, with $k$ machines receiving $L_-$ intervals and the remaining $m-k$ machines receiving $L_+$ intervals.
Again, loads differ by at most one, so the schedule is EF1.
The total number of accepted intervals is
\[
|\ALG(\I)| \;=\; kL_- + (m-k)L_+ \;=\; |\OPT(\I)|-(m-1).
\]
Since $|\OPT(\I)|\ge 3m$, we have $m-1\le \tfrac{1}{3}|\OPT(\I)|$, and therefore
\[
|\ALG(\I)| \;\ge\; |\OPT(\I)| - \tfrac{1}{3}|\OPT(\I)| \;=\; \tfrac{2}{3}|\OPT(\I)|.
\]
In all cases, the algorithm returns an EF1 schedule with value at least $\tfrac{2}{3}|\OPT(\I)|$.
\end{proof}
\paragraph{Extension to Multiple Weight Levels.}
We now show that the offline EF1 guarantee extends beyond uniform valuation to settings with a bounded number of distinct weights, at the cost of a proportional loss in efficiency.
\begin{corollary}[Offline EF1 with Multiple Weight Levels]
\label{cor:offline-ef1-multiple-weights}
In the \emph{offline} setting with $m$ identical machines, suppose that interval weights take values from a set of $k$ distinct positive weights.
Then there exists a polynomial-time EF1 scheduling algorithm whose price of fairness is at most $\tfrac{3}{2}\,k$.
That is, for every instance $\I$,
\[
\frac{|\OPT(\I)|}{|\ALG(\I)|} \;\le\; \tfrac{3}{2}\,k.
\]
\end{corollary}
\begin{proof}
    Let the set of distinct weights be $\{\omega_1,\dots,\omega_k\}$, and for each $t\in[k]$ let
    \[
    \I_t \;=\; \{\, I\in\I : w(I)=\omega_t \,\}
    \]
    denote the sub-instance consisting of intervals of weight $\omega_t$.
    For each $t$, we apply the algorithm of Theorem~\ref{thm:offline-ef1-upper} to $\I_t$, obtaining an EF1 schedule $\ALG_t$.
    Since all intervals in $\I_t$ share the same weight $\omega_t$, applying Theorem~\ref{thm:offline-ef1-upper} to $\I_t$ viewed as an unweighted instance and multiplying by $\omega_t$ yields
    \[
    w(\OPT(\I_t)) \;\le\; \tfrac{3}{2}\,w(\ALG_t).
    \]
    Let $t^\star\in[k]$ be an index maximizing $w(\ALG_t)$, and output $\ALG_{t^\star}$.
    The resulting schedule is EF1.
    To relate this to the global offline optimum, observe that for each $t$, the set of intervals of weight $\omega_t$ selected by $\OPT(\I)$ forms a feasible schedule for the sub-instance $\I_t$. Hence
    \[
    w(\OPT(\I))
    \;\le\; \sum_{t=1}^k w(\OPT(\I_t)).
    \]
    Therefore,
    \[
    w(\OPT(\I))
    \;\le\;
    \sum_{t=1}^k w(\OPT(\I_t))
    \;\le\;
    \sum_{t=1}^k \tfrac{3}{2}\,w(\ALG_t)
    \;\le\;
    \tfrac{3}{2}\,k\cdot w(\ALG_{t^\star}).
    \]
    This proves that the price of fairness is at most $\tfrac{3}{2}k$.
\end{proof}
\subsection{Lower Bounds}
We now show that the $3/2$ price of fairness established in the previous subsection is unavoidable.
Specifically, we prove that no EF1 scheduling algorithm can achieve a better approximation guarantee with respect to the offline optimum without fairness.
Our lower bounds hold in two fundamental settings: the unweighted case, and the unit-length weighted case, where intervals may have arbitrary weights but identical lengths.
Together, these results show that the efficiency loss arises from the fairness constraint itself, rather than from specific features of the valuation model or interval lengths.
\begin{restatable}[Offline EF1 Lower Bound: Unweighted]{theorem}{OfflineLowerBoundUnitWeight}
In the \emph{unweighted offline} setting, no EF1 scheduling algorithm can achieve a price of fairness strictly smaller than $\tfrac{3}{2}$. Specifically, for every $m\ge 2$ there is an instance $\I$ on $m$ identical machines such that every EF1 schedule $S$ of $\I$ satisfies $|\OPT(\I)|\ge \frac{3m-2}{2m-1}\,|S|$.
\label{thm:offline-ef1-lower}
\end{restatable}
\begin{proof}
Fix $m\ge 2$. Consider the instance $\I$ consisting of the following intervals, all of unit weight.
\begin{itemize}
    \item $2m-1$ \emph{short} intervals of unit length: for each $t\in\{0,1,\ldots,2m-2\}$ there is one interval
    \[
    I_t \;=\; [t,t+1).
    \]
    \item $m-1$ \emph{long} intervals of length $2m-1$: for each $q\in\{1,\ldots,m-1\}$ there is one interval
    \[
    L_q \;=\; [0,2m-1).
    \]
\end{itemize}
\paragraph{Benchmark (offline optimum without fairness).}
An offline optimum without fairness schedules all $3m-2$ intervals: place the $m-1$ long intervals on $m-1$ distinct machines, and place all $2m-1$ short intervals sequentially on the remaining machine. Hence
\[
|\OPT(\I)| \;=\; (m-1) + (2m-1) \;=\; 3m-2.
\]
\paragraph{Any EF1 schedule accepts at most $2m-1$ intervals.}
Recall from \Cref{obs:ef1-loads} that in the unweighted setting, a schedule is EF1 if and only if the loads of any two machines differ by at most one.
Suppose first that an EF1 schedule accepts at least one long interval. Since each long interval overlaps every other interval of the instance, the machine carrying it has load exactly~$1$. By \Cref{obs:ef1-loads}, every machine then has load at most~$2$, so the total number of accepted intervals is at most $1+2(m-1)=2m-1$.
If, on the other hand, the EF1 schedule accepts no long interval, then it accepts only short intervals, and since there are only $2m-1$ short intervals, the total is again at most $2m-1$.
Therefore every EF1 schedule $S$ of $\I$ satisfies $|S| \le 2m-1$.
\paragraph{Price of fairness.}
Combining the two bounds, every EF1 schedule $S$ of $\I$ satisfies
\[
\frac{|\OPT(\I)|}{|S|}
\;\ge\;
\frac{3m-2}{2m-1}.
\]
Since this quantity tends to $\tfrac{3}{2}$ as $m$ grows, no EF1 scheduling algorithm can achieve a price of fairness strictly smaller than $\tfrac{3}{2}$.
\end{proof}
\begin{figure}[!t]
  \centering
  \begin{subfigure}[t]{0.48\textwidth}
    \centering
    \begin{tikzpicture}[x=0.6cm,y=0.8cm]
      \def\Xmax{7}
      \def\eps{0.10}
      \machineRow{4}{Machine 1}{\Xmax}
      \machineRow{3}{Machine 2}{\Xmax}
      \machineRow{2}{Machine 3}{\Xmax}
      \machineRow{1}{Machine 4}{\Xmax}
      \drawInterval[MyRed]{solid}{0+\eps}{7-\eps}{3}
      \drawInterval[MyRed]{solid}{0+\eps}{7-\eps}{2}
      \drawInterval[MyRed]{solid}{0+\eps}{7-\eps}{1}
      \foreach \t in {0,...,6} {
        \drawInterval[MyRed]{solid}{\t+\eps}{\t+1-\eps}{4}
      }
    \end{tikzpicture}
    \caption{Offline optimum without fairness, accepting all $3m-2$ intervals.}
    \label{fig:lb-sub-no-fair}
  \end{subfigure}\hfill
  \begin{subfigure}[t]{0.48\textwidth}
    \centering
    \begin{tikzpicture}[x=0.6cm,y=0.8cm]
      \def\Xmax{7}
      \def\eps{0.10}
      \machineRow{4}{Machine 1}{\Xmax}
      \machineRow{3}{Machine 2}{\Xmax}
      \machineRow{2}{Machine 3}{\Xmax}
      \machineRow{1}{Machine 4}{\Xmax}
      \drawInterval[MyBlue]{solid}{0+\eps}{1-\eps}{4}
      \drawInterval[MyBlue]{solid}{4+\eps}{5-\eps}{4}
      \drawInterval[MyBlue]{solid}{1+\eps}{2-\eps}{3}
      \drawInterval[MyBlue]{solid}{5+\eps}{6-\eps}{3}
      \drawInterval[MyBlue]{solid}{2+\eps}{3-\eps}{2}
      \drawInterval[MyBlue]{solid}{6+\eps}{7-\eps}{2}
      \drawInterval[MyBlue]{solid}{3+\eps}{4-\eps}{1}
    \end{tikzpicture}
    \caption{An EF1 schedule, which can accept at most $2m-1$ intervals.}
    \label{fig:lb-sub-ef1}
  \end{subfigure}
  \caption{Lower-bound instance for Theorem~\ref{thm:offline-ef1-lower} in the unweighted offline setting, illustrated for $m=4$.
The offline optimum without fairness schedules all $3m-2$ intervals, whereas any EF1 schedule can accept at most $2m-1$ intervals; one such EF1 schedule is shown.}
  \label{fig:lb-two-subfigs}
\end{figure}
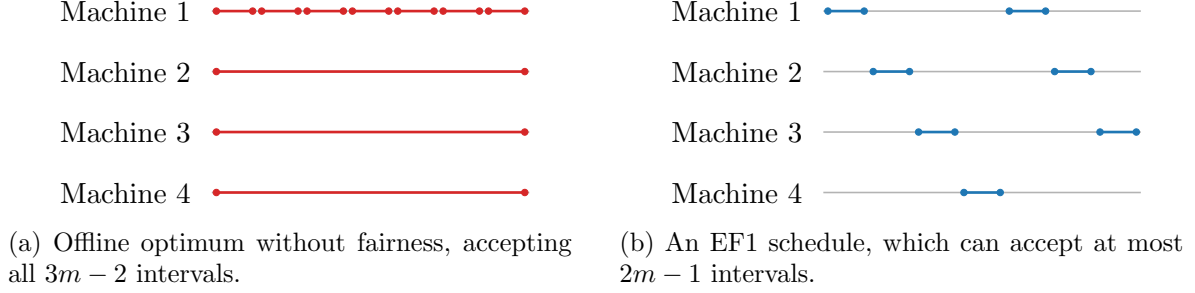
We next show that this lower bound persists even in a more structured setting, where all intervals have identical length but arbitrary weights.
\begin{restatable}[Offline EF1 Lower Bound: Unit-Length]{theorem}{OfflineLowerBoundUnitLength}
In the \emph{unit-length offline} setting, where intervals have identical lengths and arbitrary nonnegative weights, no EF1 scheduling algorithm can achieve a price of fairness strictly smaller than $\tfrac{3}{2}$.
\label{thm:unitlength-offline-ef1-lower}
\end{restatable}
\begin{proof}
Fix $\varepsilon\in(0,1)$ and an even number $m\ge 2$ of machines. Since the price of fairness is a worst-case guarantee over all instances, it suffices to exhibit instances on an even number of machines.
Partition the machines into $m/2$ disjoint pairs.
For each pair $p\in\{1,\dots,m/2\}$, we create three unit-length intervals
\[
I_1^p=[0,1),\qquad
I_2^p=\Bigl[\tfrac{1}{2},\tfrac{3}{2}\Bigr),\qquad
I_3^p=[1,2),
\]
with weights $w(I_1^p)=w(I_3^p)=1$ and $w(I_2^p)=1-\varepsilon$.
Let $\I$ be the union of all these intervals.
\paragraph{Benchmark (offline optimum without fairness).}
Fix a pair $p$ and consider any feasible schedule that accepts all three intervals $\{I_1^p,I_2^p,I_3^p\}$ using the two machines of that pair.
Since $I_2^p$ overlaps both $I_1^p$ and $I_3^p$, any machine that runs $I_2^p$ cannot run either neighbor.
Therefore, to accept all three intervals, one must schedule $I_2^p$ on one machine and schedule both $I_1^p$ and $I_3^p$ (which do not overlap) on the other machine of the pair.
Hence, within each pair $p$, the unconstrained optimum value is $w(I_1^p)+w(I_2^p)+w(I_3^p)=3-\varepsilon$.
By scheduling each pair independently on its two machines, we obtain
\[
w(\OPT(\I)) \;\ge\; \frac{m}{2}\,(3-\varepsilon).
\]
\paragraph{EF1 upper bound.}
We show that any EF1 schedule has total weight at most $m$.
Consider first the case in which some interval of type $I_2$ is accepted, say $I_2^p$.
Since every interval of type $I_2$ conflicts with all other intervals, including all intervals of types $I_1$ and $I_3$, the machine containing $I_2^p$ receives total weight exactly $1-\varepsilon$.
On the other hand, no other machine can receive two accepted intervals of types $I_1$ and $I_3$: such a machine would have total weight $2$, and even after removing either one of its intervals, its remaining weight would be $1>1-\varepsilon$, contradicting EF1 with respect to the machine containing $I_2^p$.
Therefore, every machine receives weight at most $1$, except possibly the machines containing intervals of type $I_2$, whose weight is $1-\varepsilon$.
Thus the total weight of any EF1 schedule is at most $m$.
It remains to consider the case in which no interval of type $I_2$ is accepted.
Then the schedule can only use intervals of types $I_1$ and $I_3$.
Since there are $m/2$ pairs and each pair contributes at most two such intervals of unit weight, the total weight is at most
\[
\frac{m}{2}\cdot 2 = m.
\]
Hence, in all cases, every EF1 schedule $S$ of $\I$ satisfies
\[
w(S) \le m.
\]
\paragraph{Price of fairness.}
Combining the bounds, every EF1 schedule $S$ of $\I$ satisfies
\[
\frac{w(\OPT(\I))}{w(S)}
\;\ge\;
\frac{\frac{m}{2}(3-\varepsilon)}{m}
\;=\;
\frac{3-\varepsilon}{2}.
\]
Letting $\varepsilon\to 0$ yields the lower bound $\tfrac{3}{2}$.
\end{proof}
The lower-bound instances for Theorems~\ref{thm:offline-ef1-lower} and~\ref{thm:unitlength-offline-ef1-lower} are illustrated in Figures~\ref{fig:lb-two-subfigs} and~\ref{fig:unitlen-ef1-lb-figure}, respectively, for $m=4$ machines. Together, these theorems show that the $3/2$ price of fairness is intrinsic to EF1 and cannot be avoided even under strong structural restrictions on the intervals.
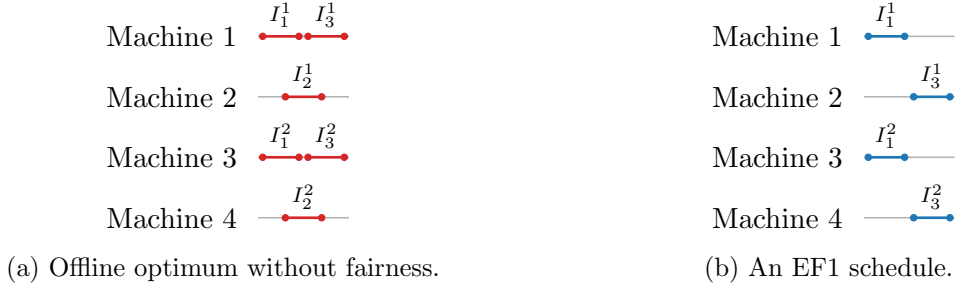
\begin{figure}[!t]
  \centering
  \begin{subfigure}[t]{0.485\textwidth}
    \centering
    \begin{tikzpicture}[x=0.6cm,y=0.8cm]
      \def\Xmax{2}
      \def\eps{0.10}
      \machineRow{4}{Machine 1}{\Xmax}
      \machineRow{3}{Machine 2}{\Xmax}
      \machineRow{2}{Machine 3}{\Xmax}
      \machineRow{1}{Machine 4}{\Xmax}
      \drawInterval[MyRed]{solid}{0+\eps}{1-\eps}{4}
      \node[font=\scriptsize] at (0.5,4.35) {$I_1^1$};
      \drawInterval[MyRed]{solid}{1+\eps}{2-\eps}{4}
      \node[font=\scriptsize] at (1.5,4.35) {$I_3^1$};
      \drawInterval[MyRed]{solid}{0.5+\eps}{1.5-\eps}{3}
      \node[font=\scriptsize] at (1.0,3.35) {$I_2^1$};
      \drawInterval[MyRed]{solid}{0+\eps}{1-\eps}{2}
      \node[font=\scriptsize] at (0.5,2.35) {$I_1^2$};
      \drawInterval[MyRed]{solid}{1+\eps}{2-\eps}{2}
      \node[font=\scriptsize] at (1.5,2.35) {$I_3^2$};
      \drawInterval[MyRed]{solid}{0.5+\eps}{1.5-\eps}{1}
      \node[font=\scriptsize] at (1.0,1.35) {$I_2^2$};
    \end{tikzpicture}
    \caption{Offline optimum without fairness.}
    \label{fig:unitlen-lb-opt}
  \end{subfigure}\hfill
  \begin{subfigure}[t]{0.485\textwidth}
    \centering
    \begin{tikzpicture}[x=0.6cm,y=0.8cm]
      \def\Xmax{2}
      \def\eps{0.10}
      \machineRow{4}{Machine 1}{\Xmax}
      \machineRow{3}{Machine 2}{\Xmax}
      \machineRow{2}{Machine 3}{\Xmax}
      \machineRow{1}{Machine 4}{\Xmax}
      \drawInterval[MyBlue]{solid}{0+\eps}{1-\eps}{4}
      \node[font=\scriptsize] at (0.5,4.35) {$I_1^1$};
      \drawInterval[MyBlue]{solid}{1+\eps}{2-\eps}{3}
      \node[font=\scriptsize] at (1.5,3.35) {$I_3^1$};
      \drawInterval[MyBlue]{solid}{0+\eps}{1-\eps}{2}
      \node[font=\scriptsize] at (0.5,2.35) {$I_1^2$};
      \drawInterval[MyBlue]{solid}{1+\eps}{2-\eps}{1}
      \node[font=\scriptsize] at (1.5,1.35) {$I_3^2$};
    \end{tikzpicture}
    \caption{An EF1 schedule.}
    \label{fig:unitlen-lb-ef1}
  \end{subfigure}
  \caption{Lower-bound instance for Theorem~\ref{thm:unitlength-offline-ef1-lower} in the unit-length offline setting, illustrated for $m=4$.
  Partition the machines into $m/2$ disjoint pairs. The offline optimum without fairness accepts all three intervals in each pair, whereas any EF1 schedule can accept at most two intervals per pair; one such EF1 schedule is shown.}
  \label{fig:unitlen-ef1-lb-figure}
\end{figure}

\section{Online EF1 Scheduling}
We now turn to the online setting, where intervals arrive sequentially and decisions must be made without knowledge of future inputs.
Recall that when intervals can have arbitrary weights, no constant-factor competitive guarantees are possible even on a single machine without fairness constraints~\citep{WOEGINGER19945}. This motivates the consideration of the \emph{unweighted} setting where all intervals have unit weight.
In the online model, enforcing fairness introduces an additional source of difficulty beyond feasibility, since EF1 must be maintained at all times under irrevocable rejections and limited foresight.
Our goal is to understand how much additional efficiency is lost when EF1 is imposed in the presence of online uncertainty.
We show that this loss can be bounded by a constant factor, and we provide a matching lower bound establishing the tightness of our guarantees.
A useful term to consider is the \emph{leader} of a machine. Given a partial schedule, the leader of a machine is the rightmost interval assigned to it. Note that, given the schedule, the leader is uniquely specified, and if there is no interval on a machine, we say that the leader is $\emptyset$. We will also find it useful to call a machine \emph{available} at time $t$ if an interval starting at time $t$ can be assigned to the machine without creating any conflict. Otherwise, we will call the machine \emph{busy}.
\paragraph{Greedy-Optimal Algorithm (Without Fairness).}
In the unweighted setting, the offline optimum can be achieved even in the online model.
As shown by \citet{faigle_note_1995}, there exists an online algorithm that maximizes the number of accepted intervals; we refer to this algorithm as \emph{Greedy-Optimal} and denote it by $\E$.
The Greedy-Optimal algorithm is a simple greedy procedure with revocation.
When an interval $I_i=[s_i,e_i)$ arrives, if it can be assigned to a machine without causing a conflict, it is accepted.
Otherwise, every machine's leader overlaps $I_i$, and the algorithm compares $e_i$ with the largest end time among the leaders. If $e_i$ is smaller, the algorithm revokes the leader with the largest end time and assigns $I_i$ to that machine; otherwise, $I_i$ is rejected.
This rule guarantees that the number of accepted intervals is maximized and thus achieves optimal efficiency.
\subsection{Online EF1 Algorithm}
We now present our fair online algorithm, called \emph{Greedy-Balanced}.
The algorithm is designed for the unweighted setting and builds directly on the Greedy-Optimal benchmark algorithm described above.
At a high level, Greedy-Balanced follows the same acceptance decisions as Greedy-Optimal whenever possible, but modifies the assignment of intervals to machines in order to maintain EF1 at all times.
This balancing step may require revoking previously accepted intervals, but only when necessary to preserve fairness.
We show that Greedy-Balanced maintains EF1 throughout the execution and achieves a fair competitive ratio of $2-\tfrac{1}{m}$ with respect to the offline optimum without fairness.
This guarantee is tight, as we establish in the subsequent lower bound subsection.
\paragraph{Greedy-Balanced Algorithm.}
Let $\E$ be the Greedy-Optimal online algorithm of \citet{faigle_note_1995}. Our Greedy-Balanced algorithm runs a parallel execution of $\E$ and only considers intervals that are accepted by $\E$ for scheduling.
For each arriving interval $I_i=[s_i,e_i)$, the algorithm proceeds as follows.
At a high level, our algorithm proceeds in four steps. First, it filters intervals according to $\E$ and only considers those accepted by $\E$. Second, it checks if the interval can be assigned to a feasible machine while preserving EF1. Specifically, the algorithm considers the set of available machines at time $s_i$ that are \emph{light} (i.e., the machines with minimum load) and places the interval on an arbitrary available light machine whenever such a placement exists. Third, if no such assignment is possible, the algorithm attempts to repair the current schedule by replacing a previously accepted overlapping interval with the new one.
In this case, among the light machines, it revokes the leader with the largest
end time and schedules $I_i$ in its place (provided that this leader overlaps
$I_i$ and ends after $e_i$). Fourth, if neither assignment nor replacement is possible, the interval is rejected.
Formally, the algorithm proceeds as follows.
\begin{enumerate}
    \item  \textbf{Filter by $\E$:}
    If $\E$ rejects $I_i$, then our algorithm also rejects $I_i$.
    If $\E$ revokes an interval $J$ to accept $I_i$ and our algorithm is also currently running $J$ on some machine, then it also revokes $J$ and schedules $I_i$ on the same machine.
    Otherwise, it proceeds to the next step.

     \item \textbf{EF1-preserving placement:}
    Let $\mathcal{F}_i \subseteq \A$ be the set of machines on which $I_i$ can be scheduled at time $s_i$ without creating a conflict, and let $C_\ell$ denote the current number of accepted intervals assigned to machine $\ell$.
    If $\mathcal{F}_i\neq\emptyset$, choose
    \[
    f \in \arg\min_{\ell\in\mathcal{F}_i} C_\ell
    \]
    where ties can be broken arbitrarily. If assigning $I_i$ to $f$ preserves EF1, then assign $I_i$ to $f$ and continue to the next arrival.

     \item \textbf{Repair via replacement on a light machine:}
    If assigning $I_i$ to any machine in $\mathcal{F}_i$ violates EF1, let
    \[
    L_i=\{k\in\A : C_k=\min_{\ell\in\A} C_\ell\}
    \]
    be the set of light machines.
    If there exists a machine $k\in L_i$ and an interval $J\in S_k$ such that $J$ overlaps $I_i$ and $e_J>e_i$, then replace $J$ with the interval $I_i$ on the machine $k$.
    Among all such pairs $(k,J)$, choose one with maximum $e_J$.

    \item \textbf{Otherwise, reject $I_i$:}
    If no such replacement is possible, reject $I_i$.
\end{enumerate}
We now state the performance guarantee of the Greedy-Balanced algorithm.
\begin{restatable}[Online EF1 Upper Bound: Unweighted]{theorem}{OnlineUpperBoundUnitWeight}
In the \emph{unweighted online} setting with $m$ identical machines, the \emph{Greedy-Balanced} algorithm maintains EF1 at all times and has fair competitive ratio at most $2-\tfrac{1}{m}$.
That is, for every input instance $\I$,
\[
\frac{|\OPT(\I)|}{|\ALG(\I)|} \;\le\; 2-\tfrac{1}{m}.
\]
\label{thm:gb-ef1-competitiveness}
\end{restatable}
\begin{proof}
Recall that the Greedy-Optimal algorithm $\E$ is optimal in the unweighted online setting, and therefore
\[
|\E(\I)| = |\OPT(\I)|
\]
for every instance $\I$. It thus suffices to prove that
\[
\frac{|\E(\I)|}{|\ALG(\I)|} \le 2-\frac{1}{m}.
\]
In what follows, we identify $\E(\I)$ with the final set of intervals accepted by $\E$ on the instance $\I$; that is, an interval belongs to $\E(\I)$ if and only if $\E$ accepts it upon arrival and never revokes it.
We assume that all $s_j$ and $e_j$ are distinct and strictly positive. This is without loss of generality since although the resulting schedule may differ, the number of accepted intervals would be the same.
The next observation follows from the description of the algorithm.
\begin{observation}
The \emph{Greedy-Balanced} algorithm is EF1 at all times. Furthermore, the loads are non-decreasing with time.
\label{obs:EF1-NonDecreasingLoads}
\end{observation}
Note that, by the above observation together with \Cref{obs:ef1-loads}, at any time point $t$, there exists an integer $\alpha$ such that each machine has a load of $\alpha$ or $\alpha +1$.
Furthermore, before any machine can reach load $\alpha+2$, every machine must
first reach load $\alpha+1$. Therefore, at each time $t$, we call a machine \emph{light} if its load at time $t$ is minimum among all machines, and \emph{heavy} otherwise; that is, the light machines are those with load $\alpha$ and the heavy machines are those with load $\alpha+1$. Note that both notions depend on the time $t$: a machine that is heavy at some time becomes light again once all other machines have caught up.
This motivates a natural indexing of the timeline into \emph{blocks} as follows; see \Cref{fig:blocks} for an illustration.
For convenience, let $t_0:=0$, and for each $k\ge 1$, let $t_k$ be the first time at which every machine has load exactly $k$. Let $\omega\ge 0$ be the largest $k$ such that $t_k$ is finite. For each $k=0,\ldots,\omega-1$, we define the $k$-th \emph{block} by
\[
B_k := [t_k,t_{k+1}),
\]
and we define the final block by $B_\omega := [t_\omega,\infty)$.
Note that the interval that brings a machine from load $k$ to $k+1$ for the first time might be revoked later. Recall that revocation does not decrease the load of a machine.
\begin{figure}[t]
\centering
\begin{tikzpicture}[x=1.4cm,y=0.95cm,>=stealth,
    event/.style={circle, fill=black, inner sep=1.5pt},
]
\fill[gray!12] (0,-0.5)   rectangle (2.0,2.55);
\fill[gray!24] (2.0,-0.5) rectangle (4.4,2.55);
\fill[gray!12] (4.4,-0.5) rectangle (6.6,2.55);
\fill[gray!24] (6.6,-0.5) rectangle (7.7,2.55);
\node[font=\small] at (1.0,2.85)  {$B_0$};
\node[font=\small] at (3.2,2.85)  {$B_1$};
\node[font=\small] at (5.5,2.85)  {$B_2$};
\node[font=\small] at (7.15,2.85) {$B_3\cdots$};
\foreach \y/\name in {2/M_1,1/M_2,0/M_3}{
  \draw[gray!60] (0,\y) -- (7.7,\y);
  \node[left, font=\small] at (-0.15,\y) {$\name$};
}
\draw[line width=2.4pt, orange!85!black] (0.5,2) -- (2.0,2);
\draw[line width=2.4pt, orange!85!black] (1.2,1) -- (2.0,1);
\draw[line width=2.4pt, orange!85!black] (2.6,1) -- (4.4,1);
\draw[line width=2.4pt, orange!85!black] (3.5,2) -- (4.4,2);
\draw[line width=2.4pt, orange!85!black] (5.0,0) -- (6.6,0);
\draw[line width=2.4pt, orange!85!black] (5.8,1) -- (6.6,1);
\foreach \x/\y/\l in {0.5/2/1, 1.2/1/1, 2.0/0/1,
                      2.6/1/2, 3.5/2/2, 4.4/0/2,
                      5.0/0/3, 5.8/1/3, 6.6/2/3}{
  \node[event] at (\x,\y) {};
  \node[font=\scriptsize, above=2.5pt] at (\x,\y) {$\l$};
}
\foreach \x/\t in {0/{t_0=0}, 2.0/{t_1}, 4.4/{t_2}, 6.6/{t_3}}{
  \draw[dashed] (\x,-0.5) -- (\x,2.55);
  \node[below, font=\small] at (\x,-0.55) {$\t$};
}
\draw[->] (0,-0.5) -- (7.95,-0.5) node[right, font=\small] {time};
\begin{scope}[shift={(0.2,-1.5)}]
  \node[event] at (0,0) {};
  \node[right, font=\small] at (0.12,0) {acceptance event (label $=$ new load)};
  \draw[line width=2.4pt, orange!85!black] (4.55,0) -- (4.95,0);
  \node[right, font=\small] at (5.05,0) {heavy};
  \draw[gray!60] (6.25,0) -- (6.65,0);
  \node[right, font=\small] at (6.75,0) {light};
\end{scope}
\end{tikzpicture}
\caption{Illustration of the blocks for $m=3$ machines. Each dot marks an acceptance event, labeled with the resulting load of the machine; recall that revocations do not change loads, so loads are non-decreasing. For each $k\ge 1$, $t_k$ is the first time at which every machine has load exactly $k$, and $B_k=[t_k,t_{k+1})$. Within block $B_k$, a machine is \emph{light} (thin line) if its load is $k$ and \emph{heavy} (thick line) if its load is $k+1$; each block ends with the acceptance event that brings the last machine to load $k+1$, at which point every machine is light with respect to the next block. In particular, for every $k<\omega$, exactly $m$ acceptance events occur in $(t_k,t_{k+1}]$, i.e., $A_k=m$.}
\label{fig:blocks}
\end{figure}
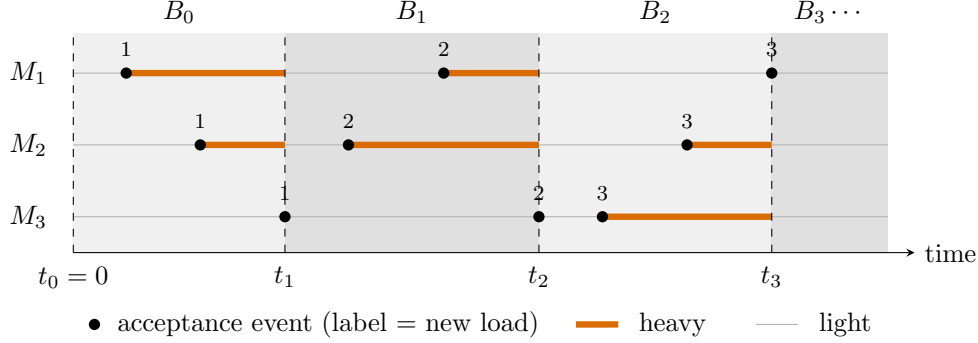
Call the start time $s_i$ an \emph{acceptance event} if the interval is assigned in step~(2), and a \emph{rejection event} if it revokes an interval in step~(3) or if it gets rejected in step~(4).
Let $A_k$ be the number of acceptance events in $(t_k,t_{k+1}]$ for each $k<\omega$, and let $A_\omega$ be the number of acceptance events in $(t_\omega,\infty)$ (so that the acceptance event occurring at time $t_{k+1}$, which completes the block, is counted towards $B_k$; rejection events occur at none of the times $t_k$, so no such convention is needed for them).
Observe that for every $0\le k<\omega$, we have $A_k=m$, since block $B_k$ ends exactly when every machine has received one more accepted interval. Furthermore, there are no rejection events in block $B_0$, since at time $t_0=0$ there are no intervals assigned to any machine, so the first $m$ intervals accepted by $\E$ are also accepted by $\ALG$ and assigned to distinct machines.
With every rejection event we associate its \emph{trigger}, its \emph{lost interval}, and its \emph{receiver}, as follows. Consider a rejection event at time $s_i$, caused by an arriving interval $I_i=[s_i,e_i)$; we call $I_i$ the trigger of the event. If the event occurs in step~(3), then \ALG{} revokes some interval $J$ from a machine $r$ in order to accept $I_i$; in this case, the lost interval of the event is $J$, and its receiver is the machine $r$. Otherwise, the event occurs in step~(4) and \ALG{} rejects $I_i$; in this case, the lost interval is $I_i$ itself, and the receiver is a dummy node $\bot$. Note that, in either case, the lost interval contains $s_i$ and ends no earlier than the trigger: for a step~(3) event this holds by the condition $e_J>e_i$ of step~(3), and for a step~(4) event the lost interval is the trigger.
The remainder of the proof is organized as follows. After recording a few structural facts about the two schedules (\Cref{obs:mirror}, \Cref{obs:rejection-structure}, \Cref{obs:leader-decay}, and \Cref{claim:no-late-rejections}), we single out the rejection events that actually hurt the algorithm, called \emph{loss events}, at which \ALG{} loses
an interval that $\E$ keeps until the end. The Provider Lemma (\Cref{lem:providers}) associates with each loss event a heavy machine, its \emph{provider}, which records spare capacity that the fairness constraint prevented \ALG{} from using. We then build an auxiliary graph $G_k$ on the machines, in which every loss event of block $B_k$ is captured by an edge from its provider to its receiver. Finally, we charge each edge of $G_k$ to one of the machines, in such a way that every machine is charged at most once and at least one machine remains uncharged; hence each block contains at most $m-1$ loss events, and the theorem follows.
Since \ALG{} schedules an interval only if $\E$ accepts it upon arrival, and mirrors the revocations of $\E$ in step~(1), the schedule of \ALG{} is at all times contained in the current schedule of $\E$.
\begin{observation}\label{obs:mirror}
At every point in time, every interval currently scheduled by \ALG{} is currently accepted by $\E$.
\end{observation}
\begin{observation}\label{obs:rejection-structure}
Consider a rejection event at time $s_i$ with trigger $I_i$ and lost interval $J$. Then every light machine is busy at time $s_i$ (equivalently, every available machine is heavy), and the leader $J_\ell$ of every light machine $\ell$ satisfies $e_{J_\ell}\le e_J$.
\end{observation}
\begin{observation}\label{obs:leader-decay}
Consider a machine that remains light throughout a time window contained in a block. Then the end time of its leader can only decrease over this window.
\end{observation}
\begin{claim}[No late rejection events]\label{claim:no-late-rejections}
Let $s$ be a rejection event in block $B_k$, and let $e^\star$ be an upper bound on the end times of the leaders of all machines that are light at time $s$. Then no rejection event occurs in $B_k$ after time $e^\star$.
\end{claim}
\begin{proof}
Suppose that a rejection event occurs at a time $s'\in B_k$ with $s'>e^\star$. Since $s'\in B_k$, not every machine has load $k+1$ at time $s'$ (for $k<\omega$ this would give $t_{k+1}\le s'$, contradicting $s'\in B_k$; for the final block it would make $t_{\omega+1}$ finite, contradicting the maximality of $\omega$). Hence some machine $\ell$ is light at time $s'$. Since loads never decrease and all loads in $B_k$ are $k$ or $k+1$, machine $\ell$ is light throughout $[s,s']$. By \Cref{obs:leader-decay}, the leader of $\ell$ at time $s'$ has end time at most its end time at $s$, which is at most $e^\star<s'$; hence $\ell$ is available at time $s'$. Thus $\ell$ is light and available at the rejection event $s'$, contradicting \Cref{obs:rejection-structure}.
\end{proof}
Applying \Cref{claim:no-late-rejections} with $e^\star=e_J$, which is a valid choice by \Cref{obs:rejection-structure}, yields the following consequence: the lost interval of a rejection event remains alive throughout the remaining rejection events of its block.
\begin{corollary}\label{cor:lost-alive}
Let $s<s'$ be rejection events in the same block, and let $J$ be the lost interval of the event at time $s$. Then $s'<e_J$; in particular, since $s_J<s$, the interval $J$ contains $s'$.
\end{corollary}
\paragraph{Loss events.}
We call a rejection event a \emph{loss event} if its lost interval belongs to $\E(\I)$, that is, if the event costs \ALG{} an interval that $\E$ keeps until the end. Rejection events that are not loss events cost \ALG{} only intervals that $\E$ itself eventually revokes, and our analysis will not need to count them.
The next lemma shows that the trigger of a loss event is also kept by $\E$ until the end. This fact is the backbone of the charging argument below: it allows us to compare the triggers of different loss events within the final schedule of $\E$.
\begin{lemma}\label{lem:trigger-kept}
If a rejection event loses an interval of $\E(\I)$, then $\E$ never revokes its trigger. In particular, the trigger of every loss event belongs to $\E(\I)$.
\end{lemma}
\begin{proof}
For a step~(4) event the trigger is the lost interval itself, and there is nothing to prove. Consider a step~(3) event at time $s_i$ with trigger $I_i$ that revokes an interval $J\in\E(\I)$; by the condition of step~(3), $e_i<e_J$. Note that $\E$ accepts $I_i$ upon arrival, since otherwise \ALG{} would have rejected $I_i$ in step~(1).
Suppose, towards a contradiction, that $\E$ revokes $I_i$ at the arrival of some interval at time $s'$. Since $\E$ revokes only an interval that conflicts with the arriving interval, we have $s_i<s'<e_i<e_J$. At time $s'$, the interval $J$ is scheduled by $\E$ on some machine ($\E$ accepts $J$ upon arrival and never revokes it, and $s_J<s_i<s'$). Every other interval on that machine starts before $s'<e_J$ and is disjoint from $J$, and hence ends before $s_J$; thus $J$ is the leader of its machine, with end time $e_J$. Moreover, $J$ contains $s'$ and therefore also conflicts with the arriving interval. Since $\E$ revokes the conflicting leader with the \emph{largest} end time and $e_J>e_i$, the revoked interval is not $I_i$, a contradiction.
\end{proof}
\paragraph{Providers.}
Fix a block $B_k$, and let $s^1<\cdots<s^q$ be its loss events, with triggers $T^1,\ldots,T^q$ (all in $\E(\I)$, by \Cref{lem:trigger-kept}) and lost intervals $J^1,\ldots,J^q$ (all in $\E(\I)$, by definition). The charging argument associates with each loss event a machine that ``pays'' for it, called the \emph{provider} of the event. The next lemma shows that providers can be chosen with three properties: each provider is available, and hence heavy at its event; no provider coincides with a receiver of the block; and loss events whose triggers overlap receive distinct providers. We emphasize that providers are an analysis device, chosen with hindsight: like the loss events themselves, they depend on the final set $\E(\I)$ and are not known to the algorithm at the time of the events.
\begin{lemma}[Provider Lemma]\label{lem:providers}
For each block $B_k$, one can assign to each loss event of $B_k$ a machine $p^j$, called its \emph{provider}, such that:
\begin{enumerate}
    \item $p^j$ is available, and hence heavy, at time $s^j$;
    \item no provider of a loss event of $B_k$ is the receiver of a loss event of $B_k$; and
    \item if the triggers $T^i$ and $T^j$ of two loss events overlap, then $p^i\neq p^j$; equivalently, the triggers of loss events that share a provider are pairwise non-overlapping.
\end{enumerate}
\end{lemma}
\begin{proof}
We choose the providers greedily, processing the loss events of $B_k$ in order of time. At event $j$, we choose $p^j$ to be any machine that (a) is available at time $s^j$, (b) is not the receiver of an earlier loss event of $B_k$, and (c) is not the provider $p^i$ of an earlier loss event whose trigger is still
running, i.e., satisfies $e_{T^i}>s^j$.
These rules imply the three properties. Property~1 follows from (a), since available machines are heavy at rejection events (\Cref{obs:rejection-structure}).
For property~3, if $T^i$ and $T^j$ overlap with $i<j$, then $e_{T^i}>s^j$, so rule~(c) excludes $p^i$ at event $j$. For property~2, rule~(b) excludes the receivers of all earlier events; the receiver of the event at $s^j$ is busy while $p^j$ is available; and every later receiver is light when chosen in
step~(3), while $p^j$, heavy at $s^j$ by property~1, remains heavy for the rest of the block.
It remains to show that a machine satisfying (a)--(c) always exists. Suppose that at some loss event $j$ no machine does; we exhibit $m+1$ distinct intervals that are simultaneously accepted by $\E$ and contain a common point just after $s^j$, contradicting the feasibility of the schedule of $\E$ on $m$
machines. Assign to each machine $\mu$ an interval $f(\mu)$ as follows: if $\mu$ is the receiver of an earlier loss event of $B_k$, let $f(\mu)$ be the lost interval of the \emph{first} such event; otherwise, if $\mu$ is busy at time $s^j$, let $f(\mu)$ be the interval that \ALG{} currently runs on $\mu$;
otherwise $\mu$ is available and not an earlier receiver, so by our assumption rule~(c) excludes it, i.e., $\mu=p^i$ for an earlier loss event whose trigger $T^i$ is still running, and we let $f(\mu):=T^i$ (by property~3, at most one trigger provided by $\mu$ is still running).
All these intervals, together with the current trigger $T^j$, contain the points immediately after $s^j$: first lost intervals contain $s^j$ by \Cref{cor:lost-alive}, the intervals of busy machines and still-running triggers contain $s^j$ by definition, and $T^j$ starts at $s^j$. Moreover, all of them are accepted by $\E$ just after time $s^j$: first lost intervals and triggers belong to $\E(\I)$ (by the definition of loss events and
\Cref{lem:trigger-kept}), and the intervals currently run by \ALG{} are currently accepted by $\E$ (\Cref{obs:mirror}) and are not revoked by $\E$ at the arrival of $T^j$, since otherwise step~(1) would apply and $s^j$ would not be a rejection event.
Finally, the $m+1$ intervals are pairwise distinct. Distinct receiver-case values are intervals lost at distinct events, and an interval is lost at most once; lost intervals are no longer run by \ALG{}, so they differ from the busy-case values, which are pairwise distinct as well. A trigger-case value $T^i$ differs from all busy-case and receiver-case values: if the event at $s^i$ occurred in step~(4), then $T^i$ was never scheduled by \ALG{} and never lost at a step-(3) event; if it occurred in step~(3), then $T^i$ only ever runs on the receiver $r^i$ of that event, which is a receiver-case machine, and
$T^i$ can be lost from $r^i$ only after $s^i$, whereas the first loss event received by $r^i$ occurs no later than $s^i$. Trigger-case values of distinct machines belong to distinct events and are therefore distinct. Lastly, $T^j$ arrives at time $s^j$, while all other listed intervals arrive earlier. This yields $m+1$ distinct intervals accepted by $\E$ and sharing a common point, a contradiction.
\end{proof}
\paragraph{The auxiliary graph $G_k$ of block $B_k$.}
Fix a block $B_k$, and fix a provider assignment for its loss events as in \Cref{lem:providers}. We associate with $B_k$ a directed graph $G_k$, which we call the \emph{auxiliary} graph, whose vertices are the machines together with a dummy node $\bot$, and which contains one edge per loss event of $B_k$: the edge of a loss event goes from its provider to its receiver (the dummy node $\bot$ for a step~(4) event). Let
\[
R_k := |E(G_k)|
\]
denote the number of edges of $G_k$, i.e., the number of loss events of $B_k$.
Each edge of $G_k$ is associated with a loss event, and hence with the start time of its trigger; we call this time the \emph{creation time} of the edge, and we order the edges of $G_k$ by increasing creation time.
Our goal is to bound the number of edges $R_k$. Towards this, note first that by \Cref{lem:providers} the graph $G_k$ is \emph{bipartite}: all provider nodes are contained in one part and all receiver nodes in the other. (Thus, within a block, a machine can be a provider, or a receiver, or neither.) We then use a \emph{charging} argument that charges each edge of $G_k$ to a distinct vertex (or machine), with at least one machine left uncharged; this yields $R_k\le m-1$. Subsequently, we show how this bound provides the desired bound on the fair competitive ratio of \ALG.
\begin{lemma}\label{lem:providers-receivers-disjoint}
The sets of providers and receivers in the graph $G_k$ are disjoint. Moreover, every provider is heavy from the creation time of its first outgoing edge onward.
\end{lemma}
\begin{proof}
The first part is property~2 of \Cref{lem:providers}. The second part follows from property~1 of \Cref{lem:providers} together with \Cref{obs:EF1-NonDecreasingLoads}: a provider is heavy at the creation time of each of its edges, and loads never decrease.
\end{proof}
\paragraph{Charging scheme.}
Recall that each edge of the auxiliary graph $G_k$ is directed from its provider to its receiver (possibly the dummy node $\bot$).
We charge each edge of $G_k$ either to its receiver or to its provider, according to the following rule.
\begin{quote}
    Each receiver is charged the \emph{last} incoming edge it receives in $G_k$. Every other edge (an edge entering the dummy node $\bot$, or an edge entering a receiver that later receives another incoming edge) is charged to its provider.
\end{quote}
By construction, every edge is charged to exactly one node, and every receiver is charged at most once. The main difficulty is to show that each provider is charged at most once. To this end, we establish three structural properties of $G_k$ (\Cref{claim:earlier-receivers-heavy,claim:two-incoming-kill,claim:dummy-kill}), which together imply the key fact behind the charging scheme: every edge charged to a provider is the \emph{last} edge leaving that provider (\Cref{cor:last-in-last-out,claim:dummy-kill}).
Fix a provider $p$ and consider the edges leaving $p$ in order of creation time.
We first show that whenever $p$ creates a new outgoing edge, its previous receiver must already be heavy; since step~(3) selects only light machines as receivers, that machine can never receive an edge again.
\begin{claim}[A new outgoing edge kills the previous receiver]\label{claim:earlier-receivers-heavy}
Fix a provider $p$ in $G_k$, and let
\[
p\to r_1,\ p\to r_2,\ \ldots,\ p\to r_{q-1},\ p\to r_q
\]
be the edges leaving $p$, listed in order of creation time, where each $r_j$ is a receiver machine and $r_q$ is either a receiver machine or the dummy node $\bot$. Then, for every $j<q$, the machine $r_j$ is heavy when the edge $p\to r_{j+1}$ is created, and no edge created afterwards can enter $r_j$.
\end{claim}
\begin{proof}
Fix $j<q$. Let the edge $p\to r_j$ be created by a step~(3) loss event triggered by an interval $I_j$, which \ALG{} places on machine $r_j$, and let the edge $p\to r_{j+1}$ be created by a loss event triggered by an interval $I_{j+1}$.
Since both edges leave $p$, the triggers $I_j$ and $I_{j+1}$ share the provider $p$, and are therefore non-overlapping by property~3 of \Cref{lem:providers}; hence $s_{j+1}>e_{I_j}$.
Suppose that $r_j$ is light at time $s_{j+1}$. Then $r_j$ is light throughout $[s_j,s_{j+1}]$, and $I_j$ is its leader at time $s_j$; by \Cref{obs:leader-decay}, the leader of $r_j$ at time $s_{j+1}$ ends at most at $e_{I_j}<s_{j+1}$, so $r_j$ is available. Thus $r_j$ is light and available at the rejection event $s_{j+1}$, contradicting \Cref{obs:rejection-structure}. Hence $r_j$ is heavy at time $s_{j+1}$.
Since loads never decrease within a block, $r_j$ remains heavy for the rest of the block, and since step~(3) selects only light machines as receivers, no edge created after $p\to r_{j+1}$ can enter $r_j$.
\end{proof}
The next claim shows that when a receiver of $p$ acquires a second incoming edge, which by \Cref{claim:earlier-receivers-heavy} can only happen before the next edge leaving $p$ is created, the provider $p$ is \emph{killed}: it cannot create any further edge in the block.
\begin{claim}[A new incoming edge kills the previous provider]\label{claim:two-incoming-kill}
Let $r$ be a receiver in $G_k$, and let
\[
p_1 \to r \qquad\text{and}\qquad p_2 \to r
\]
be two edges in $G_k$ such that the edge from $p_1$ is created before the edge from $p_2$, and let $s_2$ be the creation time of $p_2\to r$. Then $p_1$ cannot be the provider of any edge created after time $s_2$.
\end{claim}
\begin{proof}
    Let the edge $p_1\to r$ be created at time $s_1$ by a step~(3) loss event triggered by an interval $I_1$, which \ALG{} places on machine $r$, and let the edge $p_2\to r$ be created at time $s_2$ by a step~(3) loss event with lost interval $J$; since $J$ contains $s_2$, it is the leader of $r$ at time $s_2$.
    Since $r$ is chosen as a receiver at both times, it is light at times $s_1$ and $s_2$, and hence, by \Cref{obs:EF1-NonDecreasingLoads}, throughout $[s_1,s_2]$. Interval $I_1$ is the leader of $r$ at time $s_1$, so by \Cref{obs:leader-decay} we have $e_J\le e_{I_1}$. By \Cref{obs:rejection-structure}, the leader of every machine that is light at time $s_2$ ends at most at $e_J\le e_{I_1}$.
    Suppose, towards a contradiction, that $p_1$ is the provider of an edge of $G_k$ created at some time $s_I>s_2$, with trigger $I$. Since $I$ and $I_1$ share the provider $p_1$, they are non-overlapping by property~3 of \Cref{lem:providers}, and hence $s_I>e_{I_1}$. But then the loss event at time $s_I$ is a rejection event of $B_k$ occurring after time $e_{I_1}$, contradicting \Cref{claim:no-late-rejections} applied at the rejection event $s_2$ with $e^\star=e_{I_1}$.
\end{proof}
Combining the two claims yields the key structural property behind the charging scheme.
\begin{figure}[t]
\centering
\begin{tikzpicture}[
    x=1cm,y=1cm,
    provider/.style={circle, draw, minimum size=8mm, inner sep=0pt},
    receiver/.style={circle, draw, minimum size=8mm, inner sep=0pt},
    >=stealth
]
\node[font=\small] at (1.8,2.8) {if $p_1\to r_2$ is created};
\node[provider] (p1a) at (0,1.5) {$p_1$};
\node[provider] (p2a) at (0,0) {$p_2$};
\node[receiver] (r1a) at (3.6,1.5) {$r_1$};
\node[receiver] (r2a) at (3.6,0) {$r_2$};
\draw[->, thick] (p1a) -- (r1a);
\draw[->, thick] (p2a) -- (r2a);
\draw[->, very thick] (p1a) -- (r2a);
\draw[->, dashed] (p2a) -- (r1a);
\node[font=\small] at (8.2,2.8) {if $p_2\to r_1$ is created};
\node[provider] (p1b) at (6.4,1.5) {$p_1$};
\node[provider] (p2b) at (6.4,0) {$p_2$};
\node[receiver] (r1b) at (10,1.5) {$r_1$};
\node[receiver] (r2b) at (10,0) {$r_2$};
\draw[->, thick] (p1b) -- (r1b);
\draw[->, thick] (p2b) -- (r2b);
\draw[->, very thick] (p2b) -- (r1b);
\draw[->, dashed] (p1b) -- (r2b);
\node[font=\small] at (0,2.25) {providers};
\node[font=\small] at (3.6,2.25) {receivers};
\node[font=\small] at (6.4,2.25) {providers};
\node[font=\small] at (10,2.25) {receivers};
\end{tikzpicture}
\caption{Visual illustration of Claims~\ref{claim:earlier-receivers-heavy} and~\ref{claim:two-incoming-kill}.
Starting from edges $p_1\to r_1$ and $p_2\to r_2$, once one crossed edge is created, the other crossed edge cannot occur later. Indeed, a new outgoing edge kills the previous receiver, while a new incoming edge kills the previous provider. Dashed edges indicate configurations that cannot arise.}
\label{fig:crossed-edges-impossible}
\end{figure}
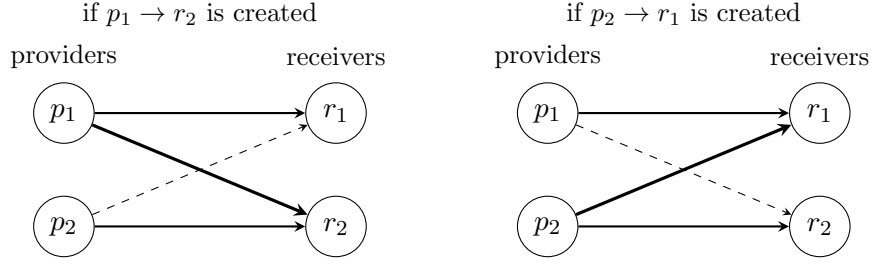
\begin{corollary}\label{cor:last-in-last-out}
If an edge $p\to r$ of $G_k$ is not the last incoming edge of its receiver $r$, then it is the last outgoing edge of its provider $p$.
\end{corollary}
\begin{proof}
Let $e'$ be an incoming edge of $r$ created after $p\to r$, say at time $s'$. Suppose, towards a contradiction, that $p$ creates an edge after $p\to r$, and let $p\to r''$ be the first such edge. By \Cref{claim:earlier-receivers-heavy}, no edge can enter $r$ after $p\to r''$ is created; hence $e'$ is created before $p\to r''$. But then \Cref{claim:two-incoming-kill} implies that $p$ cannot be the provider of any edge created after time $s'$, contradicting the existence of $p\to r''$. Hence $p\to r$ is the last edge leaving $p$.
\end{proof}
The previous claim handles the case in which a provider is killed by a later incoming edge at one of its receivers; see Figure~\ref{fig:crossed-edges-impossible} for an illustration of this interaction together with \Cref{claim:earlier-receivers-heavy}. We now consider the remaining case, in which a provider creates an edge to the dummy node $\bot$, and show that this event also kills the provider; in particular, an edge entering $\bot$ is likewise the last edge leaving its provider.
\begin{claim}[An incoming edge to $\bot$ kills its provider]\label{claim:dummy-kill}
Let $p\to \bot$ be an edge in $G_k$, and let $s$ be its creation time. Then $p$ cannot be the provider of any edge created after time $s$.
\end{claim}
\begin{proof}
    Let the edge $p\to\bot$ be created by a step~(4) loss event triggered by an interval $I_i=[s_i,e_i)$, so that $s_i=s$ and the lost interval is $I_i$ itself. By \Cref{obs:rejection-structure}, the leader of every machine that is light at time $s_i$ ends at most at $e_i$.
    Suppose, towards a contradiction, that $p$ is the provider of an edge of $G_k$ created at some time $s_I>s_i$, with trigger $I$. Since $I$ and $I_i$ share the provider $p$, they are non-overlapping by property~3 of \Cref{lem:providers}, and hence $s_I>e_i$. But then the loss event at time $s_I$ is a rejection event of $B_k$ occurring after time $e_i$, contradicting \Cref{claim:no-late-rejections} applied at the rejection event $s_i$ with $e^\star=e_i$.
\end{proof}
The previous claims establish the structural properties of $G_k$ needed for the charging argument; see Figure~\ref{fig:charging-bipartite-graph} for an illustration. In particular, every edge charged to a provider (an edge entering $\bot$, or an edge that is not the last incoming edge of its receiver) is the last edge leaving that provider, by \Cref{claim:dummy-kill} and \Cref{cor:last-in-last-out}, respectively. We can now show that each machine is charged at most once.
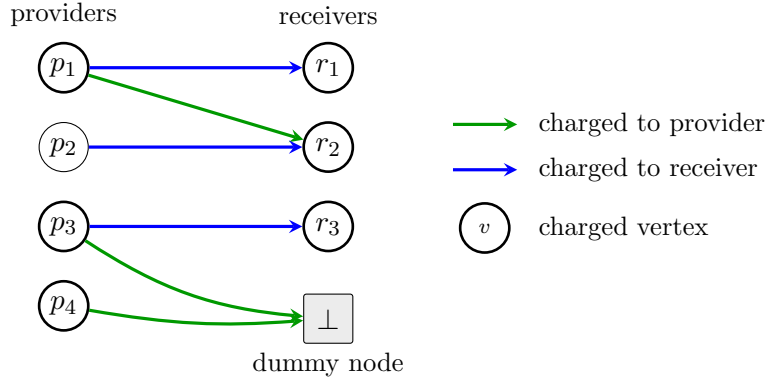
\begin{figure}[t]
\centering
\begin{tikzpicture}[
    x=1cm,y=1cm,
    provider/.style={circle, draw, minimum size=6.5mm, inner sep=0pt},
    chargedprovider/.style={circle, draw, minimum size=6.5mm, inner sep=0pt, line width=1pt, font=\bfseries},
    receiver/.style={circle, draw, minimum size=6.5mm, inner sep=0pt},
    chargedreceiver/.style={circle, draw, minimum size=6.5mm, inner sep=0pt, line width=1pt, font=\bfseries},
    dummy/.style={rectangle, draw, rounded corners=1pt, minimum width=6.5mm, minimum height=6.5mm, inner sep=1pt, fill=gray!15},
    >=stealth
]
\node[chargedprovider] (p1) at (0,2.1) {$p_1$};
\node[provider]        (p2) at (0,1.05) {$p_2$};
\node[chargedprovider] (p3) at (0,0) {$p_3$};
\node[chargedprovider] (p4) at (0,-1.05) {$p_4$};
\node[chargedreceiver] (r1) at (3.5,2.1) {$r_1$};
\node[chargedreceiver] (r2) at (3.5,1.05) {$r_2$};
\node[chargedreceiver] (r3) at (3.5,0) {$r_3$};
\node[dummy]           (d)  at (3.5,-1.22) {$\bot$};
\draw[->, very thick, green!60!black] (p1) -- (r2);
\draw[->, very thick, blue] (p2) -- (r2);
\draw[->, very thick, blue] (p1) -- (r1);
\draw[->, very thick, blue] (p3) -- (r3);
\draw[->, very thick, green!60!black] (p3) to[bend right=14] (d);
\draw[->, very thick, green!60!black] (p4) to[bend right=7] (d);
\node[font=\small] at (0,2.8) {providers};
\node[font=\small] at (3.5,2.8) {receivers};
\node[font=\small] at (3.5,-1.82) {dummy node};
\begin{scope}[shift={(5.15,0.7)}]
    \draw[->, very thick, green!60!black] (0,0.65) -- (0.85,0.65);
    \node[right, font=\small] at (1.0,0.65) {charged to provider};
    \draw[->, very thick, blue] (0,0.05) -- (0.85,0.05);
    \node[right, font=\small] at (1.0,0.05) {charged to receiver};
    \node[chargedprovider] (lp) at (0.42,-0.72) {\scriptsize $v$};
    \node[right, font=\small] at (1.0,-0.72) {charged vertex};
\end{scope}
\end{tikzpicture}
\caption{Illustration of the charging scheme on the bipartite graph $G_k$ for a block $B_k$. Providers appear on the left, receivers on the right, and the dummy node $\bot$ represents step~(4) edges. The last incoming edge of each receiver is charged to that receiver (blue). All remaining edges (earlier incoming edges and edges entering $\bot$) are charged to their providers (green); each such edge is the last edge leaving its provider. Here, $p_1\to r_2$ is created before $p_2\to r_2$; note that machine $p_2$ is uncharged.}
\label{fig:charging-bipartite-graph}
\end{figure}
\begin{lemma}\label{lem:charged-once}
Each machine is charged at most once.
\end{lemma}
\begin{proof}
    By the charging rule, a receiver is charged only its last incoming edge, and is therefore charged at most once.
    Now consider a provider $p$. An edge leaving $p$ is charged to $p$ in exactly two cases: it enters the dummy node $\bot$, or it enters a receiver $r$ and is not the last incoming edge of $r$. In the first case, the edge is the last edge leaving $p$ by \Cref{claim:dummy-kill}; in the second case, it is the last edge leaving $p$ by \Cref{cor:last-in-last-out}. Since at most one edge leaving $p$ can be its last, at most one edge is charged to $p$.
    Finally, by \Cref{lem:providers-receivers-disjoint}, no machine is both a provider and a receiver in $G_k$, so no machine can be charged once in each role. Therefore each machine is charged at most once.
\end{proof}
\begin{lemma}\label{lem:uncharged-machine}
At least one machine is uncharged.
\end{lemma}
\begin{proof}
    If some machine is neither a provider nor a receiver in $G_k$, then it is never charged, and we are done. We may therefore assume that every machine is a provider or a receiver; by \Cref{lem:providers-receivers-disjoint}, these two sets partition the machines. Suppose, towards a contradiction, that every machine is charged. Let $s$ be the creation time of the last edge of $G_k$, and note that $s<t_{k+1}$, since all edges of $G_k$ are created inside the block $B_k=[t_k,t_{k+1})$.
    We show that every machine is heavy at time $s$, which yields the desired contradiction. By \Cref{lem:providers-receivers-disjoint}, every provider is heavy from the creation of its first outgoing edge; since loads never decrease within a block, all providers are heavy at time $s$.
    Now consider any receiver $r$, and let $q\to r$ be its last incoming edge, i.e., the edge charged to $r$. We claim that $q\to r$ is not the last edge leaving $q$. Indeed, only the last edge leaving a provider can be charged to it (by \Cref{claim:dummy-kill} and \Cref{cor:last-in-last-out}); since $q\to r$ is charged to $r$, if it were the last edge leaving $q$, then no edge would be charged to $q$, contradicting our assumption. Hence $q$ creates a further edge after $q\to r$, and no later than time $s$. By \Cref{claim:earlier-receivers-heavy}, machine $r$ is heavy from the creation of that edge onward, in particular, at time $s$.
    Thus every machine is heavy at time $s$, i.e., every machine has load $k+1$, so $t_{k+1}\le s$, contradicting $s<t_{k+1}$. For the final block $B_\omega$, the same argument shows that every machine would have load $\omega+1$ at time $s$, so $t_{\omega+1}$ would be finite, contradicting the maximality of $\omega$.
\end{proof}
\begin{claim}\label{claim: m-1 loss}
For every $0\le k \le \omega$, we have $R_k \le m-1$.
\end{claim}
\begin{proof}
    By the charging scheme, every edge of $G_k$ is charged to some machine. By \Cref{lem:charged-once}, each machine is charged at most once, and by \Cref{lem:uncharged-machine}, at least one machine is uncharged. Since there are $m$ machines,
    \[
    R_k=|E(G_k)|\le m-1. \qedhere
    \]
\end{proof}
Recall that $A_k=m$ for every $0\le k<\omega$, and that $R_0=0$ since block $B_0$ contains no rejection events. Moreover, since the total load of \ALG{} increases only at acceptance events, we have $|\ALG(\I)|=\sum_{k=0}^{\omega}A_k$.
The next lemma relates $|\E(\I)|$ to the acceptance events and the edges of the auxiliary graphs: every interval that $\E$ keeps until the end is either kept by \ALG{} as well, or lost at an event that is, by definition, a loss event, and hence witnessed by an edge.
\begin{lemma}\label{lem:counting}
$\displaystyle |\E(\I)| \;\le\; \sum_{k=0}^{\omega}A_k+\sum_{k=0}^{\omega}R_k.$
\end{lemma}
\begin{proof}
Since $|\ALG(\I)|=\sum_{k}A_k$, it suffices to show that the number of intervals of $\E(\I)$ that \ALG{} does not keep until the end is at most $\sum_{k}R_k$. Let $I\in\E(\I)$ be such an interval. Since $\E$ accepts $I$ upon arrival and never revokes it, \ALG{} neither rejects nor revokes $I$ in step~(1); hence \ALG{} loses $I$ at a rejection event, by rejecting it in
step~(4) or revoking it in step~(3), and in either case $I$ is the lost interval of that event. Since $I\in\E(\I)$, this event is a loss event and therefore contributes an edge. As distinct intervals are lost at distinct rejection events, the number of such intervals is at most $\sum_{k}R_k$.
\end{proof}
By \Cref{claim: m-1 loss} and \Cref{lem:counting}, we can bound the fair competitive ratio as follows:
\begin{align*}
\frac{|\OPT(\I)|}{|\ALG(\I)|}
= \frac{|\E(\I)|}{|\ALG(\I)|}
&\le \frac{\sum_{k=0}^{\omega}A_k + \sum_{k=0}^{\omega}R_k}{\sum_{k=0}^{\omega}A_k} \\
&= 1+\frac{R_0+\sum_{k=1}^{\omega}R_k}{\sum_{k=0}^{\omega-1}A_k + A_\omega} \\
&\le 1+\frac{\omega(m-1)}{\omega m + A_\omega}
\;\le\; 1+\frac{\omega(m-1)}{\omega m}
= 2-\frac{1}{m}.
\end{align*}
This completes the proof.
\end{proof}
We next show that this fair competitive ratio is optimal by proving a matching lower bound.

\subsection{Lower Bounds}
We now show that the fair competitive ratio achieved by Greedy-Balanced is optimal among deterministic online algorithms that satisfy EF1.
Specifically, we prove that no deterministic EF1 online algorithm can achieve a fair competitive ratio strictly smaller than $2-\tfrac{1}{m}$ with respect to the offline optimum without fairness.
This lower bound shows that the additional efficiency loss observed in the online setting is unavoidable and arises from the interplay between online uncertainty and the fairness constraint.
It specifically applies to deterministic algorithms, and does not rule out the possibility of improved guarantees using randomization. Designing randomized EF1 online algorithms with better fair competitive ratios is an interesting direction for future work.

\begin{restatable}[Online EF1 Lower Bound]{theorem}{OnlineLowerBoundUnitWeight}
In the \emph{unweighted online} setting with $m$ identical machines, every deterministic online algorithm that satisfies EF1 has fair competitive ratio at least $2-\tfrac{1}{m}$.
\label{thm:lb-deterministic-ef1}
\end{restatable}

\begin{proof}
Fix any deterministic online EF1 algorithm $\ALG$ and an arbitrary machine, say machine~$1$.
We construct an adversarial input instance $\I$ consisting of $B$ \emph{recurring blocks}, where $B$ can be taken arbitrarily large.
In this instance, all intervals have unit weight.
For readability, we describe the construction with some coinciding endpoints (e.g., consecutive intervals sharing a left endpoint, or an interval starting exactly at the right endpoint of an earlier one; since intervals are half-open, the latter creates no conflict).
Strictly speaking, each start time is perturbed by a distinct, sufficiently small amount, so that the intervals of each block arrive one at a time, in the stated order; in particular, the adversary observes the decision of $\ALG$ on each interval before choosing the next one.
For a small enough perturbation, this changes no conflict between any two released intervals, and all claims below are unaffected.
We suppress the perturbation in what follows.

\paragraph{EF1 as a count-balance invariant.}
By \Cref{obs:ef1-loads}, in the unweighted setting a schedule is EF1 if and only if the loads of any two machines differ by at most one; since $\ALG$ maintains EF1 at all times, this holds at every point of its execution.
Formally, letting $C_i(t)$ be the number of intervals currently assigned to machine $i$ right after time~$t$, we have
\begin{equation}\label{eq:ef1-counts-online-lb}
\max_{i\in\A} C_i(t)-\min_{i\in\A} C_i(t)\le 1 \qquad \text{for all times $t$}.
\end{equation}

\paragraph{Recurring blocks.}
Each block starts at some time $T$ strictly after the end of all previously released intervals (so blocks do not overlap in time), and all its intervals are contained in the window $[T,T+1)$.
Within a block, the adversary maintains an \emph{active region} $[s,e)$, initialized to $[T,T+1)$, and repeats the following process.

\begin{enumerate}
\item Release the interval $I=\bigl[s,\tfrac{s+e}{2}\bigr)$, spanning the first half of the active region.
\item If $\ALG$ accepts $I$ on machine~$1$, call $I$ a \emph{success}, update the active region to $I$ (the first half of the old active region), and go back to step~1.
\item Otherwise (i.e., if $I$ is accepted on another machine or not accepted at all), call $I$ a \emph{failure}, update the active region to $\bigl[\tfrac{s+e}{2},e\bigr)$ (the part strictly after $I$), and go back to step~1.
\end{enumerate}

The block ends as soon as $m$ successes have occurred.
In addition, if the number of intervals released in a block reaches $K:=4m(B+1)$, the adversary stops releasing altogether and the instance ends; we show below that this safeguard against stalling only helps the adversary.
We emphasize that the classification is determined by the observed behavior of $\ALG$: \emph{every} interval that $\ALG$ accepts on machine~$1$ is a success, so machine~$1$ never schedules a failure.
Since the active region only shrinks, and each update places it either inside the released interval or strictly after it, we obtain the following structural property, illustrated in \Cref{fig:online-lb-illustration}.

\begin{claim}\label{claim:structure}
Within a block, every released interval is contained in every earlier success, and lies strictly after every earlier failure.
In particular, the successes of a block form a nested chain $S_1\supset S_2\supset\cdots$, ordered by release time.
\end{claim}

The design creates the following trap.
To accept a new success, machine~$1$ must revoke the previous one, since the new success is nested inside it; so out of its $m$ successes in a block, machine~$1$ keeps only the innermost, and the other $m-1$ are lost to revocations.
The optimum, in contrast, keeps the entire nested chain by spreading it over the $m$ machines, and every failure fits into its schedule as well.
Refusing to play does not help $\ALG$: if machine~$1$ stops accepting, the block never ends, the other machines can absorb only what \eqref{eq:ef1-counts-online-lb} permits, and every further release is rejected, a pure gift to the optimum.

\begin{claim}\label{claim:ALGm}
In every block, at most one released interval ever completes on machine~$1$.
\end{claim}
\begin{proof}
Machine~$1$ only ever schedules successes.
Suppose a success $S_i$ completes on machine~$1$, and consider any success $S_j$ released later in the block.
By \Cref{claim:structure}, $S_j$ is contained in the active region at its release, which is contained in $S_i$; hence $S_j$ arrives strictly before the right endpoint of $S_i$, at a moment when $S_i$ is still scheduled on machine~$1$.
Since $S_j\subset S_i$, accepting $S_j$ on machine~$1$ requires revoking $S_i$, contradicting that $S_i$ completes there.
So once a success completes on machine~$1$, no later success exists, and machine~$1$ completes at most one interval per block.
\end{proof}

\begin{claim}\label{claim:OPT2m-1}
For each block, the offline optimum without fairness accepts every released interval.
If the block completes with $m$ successes and $k$ failures, this amounts to $m+k$ intervals.
\end{claim}
\begin{proof}
Consider a block with successes $S_1\supset\cdots\supset S_q$ (where $q\le m$) and place $S_i$ on machine $i$.
For a failure $F$, let $j$ be the number of successes released before $F$, and note that $j\le m-1$, since the block ends at the $m$-th success; place $F$ on machine $j+1$.
We verify that no two intervals on a machine conflict.
Any two failures on machine $j+1$ have the same number $j$ of earlier successes, and by \Cref{claim:structure} the later one lies strictly after the earlier one.
A failure $F$ on machine $j+1$ was released before the success $S_{j+1}$, so again by \Cref{claim:structure}, $S_{j+1}$ lies strictly after $F$.
Since blocks are disjoint in time, this schedule is feasible and accepts every released interval of the instance.
\end{proof}

\paragraph{Competitive ratio lower bound.}
Assume first that all $B$ blocks complete.
Consider a block with $k$ failures, of which $\ALG$ accepts $a\le k$ on machines other than machine~$1$.
By \Cref{claim:ALGm}, the block contributes at most $1+a$ intervals to the final schedule of $\ALG$, while by \Cref{claim:OPT2m-1} it contributes $m+k\ge m+a$ intervals to the optimum; the difference is at least $m-1$.
Summing over blocks,
\[
|\OPT(\I)|\;\ge\;|\ALG(\I)|+(m-1)B .
\]
Moreover, by \Cref{claim:ALGm} the count of machine~$1$ at the end of the instance is at most $B$, so by \eqref{eq:ef1-counts-online-lb} every other machine holds at most $B+1$ intervals, and therefore
\[
|\ALG(\I)|\;\le\;B+(m-1)(B+1)\;=\;mB+m-1 .
\]
(If $|\ALG(\I)|=0$, the fair competitive ratio is unbounded, so assume $|\ALG(\I)|\ge 1$.)
Combining the two displays,
\[
\frac{|\OPT(\I)|}{|\ALG(\I)|}\;\ge\;1+\frac{(m-1)B}{mB+m-1}
\;\xrightarrow[B\to\infty]{}\;2-\frac{1}{m}.
\]
It remains to consider the case that some block hits the safeguard, i.e., releases $K=4m(B+1)$ intervals without completing.
The bound $|\ALG(\I)|\le m(B+1)$ holds at every point of the execution, by the same argument as above applied to the blocks so far.
Within the stalled block, at most $m$ released intervals are successes, and at most $(m-1)(B+1)$ are failures accepted by the other machines, since each such acceptance adds to the count of a machine other than machine~$1$, and these counts never exceed $B+1$.
Hence at least $K-m-(m-1)(B+1)\ge 2m(B+1)$ released intervals are accepted by no one, and by \Cref{claim:OPT2m-1} all of them are accepted by the optimum, so
\[
\frac{|\OPT(\I)|}{|\ALG(\I)|}\;\ge\;\frac{2m(B+1)}{m(B+1)}\;=\;2\;\ge\;2-\frac1m .
\]
In both cases, taking $B\to\infty$ completes the proof.

\end{proof}

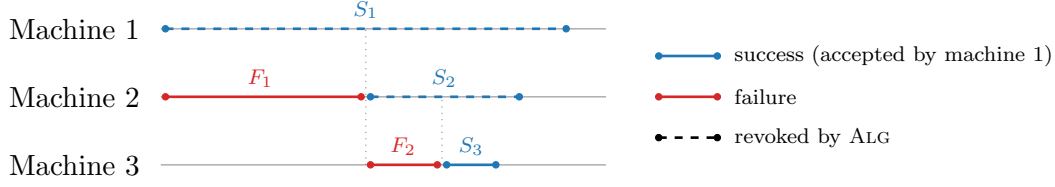
\begin{figure}[!t]
  \centering
  \begin{tikzpicture}[x=1.55cm,y=0.9cm]
    \def\Xmax{3.8}
    \def\eps{0.04}
    \machineRow{3}{Machine 1}{\Xmax}
    \machineRow{2}{Machine 2}{\Xmax}
    \machineRow{1}{Machine 3}{\Xmax}
    \drawInterval[MyBlue]{dashed}{0+\eps}{3.5-\eps}{3}
    \node[MyBlue,font=\scriptsize] at (1.75,3.28) {$S_1$};
    \drawInterval[MyRed]{solid}{0+\eps}{1.75-\eps}{2}
    \node[MyRed,font=\scriptsize] at (0.85,2.28) {$F_1$};
    \drawInterval[MyBlue]{dashed}{1.75+\eps}{3.1-\eps}{2}
    \node[MyBlue,font=\scriptsize] at (2.42,2.28) {$S_2$};
    \drawInterval[MyRed]{solid}{1.75+\eps}{2.4-\eps}{1}
    \node[MyRed,font=\scriptsize] at (2.07,1.28) {$F_2$};
    \drawInterval[MyBlue]{solid}{2.4+\eps}{2.9-\eps}{1}
    \node[MyBlue,font=\scriptsize] at (2.65,1.28) {$S_3$};
    \draw[gray,dotted] (1.75,1.0) -- (1.75,3.0);
    \draw[gray,dotted] (2.4,1.0) -- (2.4,2.0);
    \begin{scope}[shift={(4.25,0)}]
      \drawInterval[MyBlue]{solid}{0}{0.5}{2.6}
      \node[right,font=\scriptsize] at (0.56,2.6) {success (accepted by machine 1)};
      \drawInterval[MyRed]{solid}{0}{0.5}{2.0}
      \node[right,font=\scriptsize] at (0.56,2.0) {failure};
      \drawInterval[black]{dashed}{0}{0.5}{1.4}
      \node[right,font=\scriptsize] at (0.56,1.4) {revoked by $\ALG$};
    \end{scope}
  \end{tikzpicture}
  \caption{One block for $m=3$ with two failures, shown in the optimal schedule, which accepts all $2m-1=5$ intervals: $S_i$ on machine $i$ and each failure before the next success.
  The successes are nested, so accepting each new one forces machine~$1$ to revoke the previous one (dashed), and $\ALG$ keeps only the three solid intervals.
  Repeating the block yields the ratio $\tfrac{5}{3}=2-\tfrac{1}{m}$.}
  \label{fig:online-lb-illustration}
\end{figure}

\section{Experimental Evaluation}
We complement our theoretical analysis with an empirical evaluation of the Greedy-Balanced algorithm using real-world benchmark instances. 
Our theoretical guarantees focus on worst-case scenarios and are adversarial in nature, which means they may not necessarily reflect typical performance. 
In this section, we will examine how the algorithm behaves on realistic inputs and compare its empirical performance against the worst-case bounds we established earlier.

\textbf{Dataset.}
We evaluate Greedy-Balanced using benchmark instances derived from scheduling traces of computing clusters from the Parallel Workloads Archive~\citep{workloadarchive,feitelson_experience_2014,goos_benchmarks_1999}. 
This dataset has been used in previous experimental studies on interval scheduling algorithms~\citep{antoniadis2025switchingframeworkonlineinterval,10.1007/978-3-031-38906-1_14}. 
The traces record jobs scheduled on parallel processors, with each job specified by a start time, an end time, and a number of requested processors. For our analysis, we treat each job as requiring a single machine and disregard the multiplicity of processors mentioned in the original logs. 
We vary the number of machines, running each instance for $m=1,2,\ldots,300$. 
Instances with more than $10^6$ jobs are excluded from consideration.
This results in a total of $33$ instances, ranging from $\num{18238}$ to $\num{728871}$ jobs. All experiments were conducted on standard CPU-based workstation. Since our evaluation prioritizes solution quality instead of runtime, hardware specifications do not affect the results.

\textbf{Results.}
As shown in \cref{fig:plot_distribution}, Greedy-Balanced consistently outperforms its worst-case theoretical guarantees in practice. 
The maximum competitive ratio observed across all instances and machine counts is $1.306$, which occurs for the instance \texttt{UniLu-Gaia-2014-2} with $m=49$ machines. The average of the maximum competitive ratio for each instance across all machine counts is $1.235$.
Overall, these results suggest that while fairness constraints impose an inherent efficiency loss in the worst case, their practical impact on typical instances is significantly smaller.

An interesting observation is that as the number of machines increases, the competitive ratio initially grows and then starts to decline
, as shown in Figure~\ref{fig:plot_machines}. Intuitively, adding more machines introduces additional fairness constraints, which can temporarily worsen the performance. However, having more machines also provides greater flexibility in scheduling, as more machines are idle at any given time, ultimately leading to improved performance.

\begin{figure}
  \begin{minipage}[t]{0.48\textwidth}
  \includegraphics[width=\textwidth]{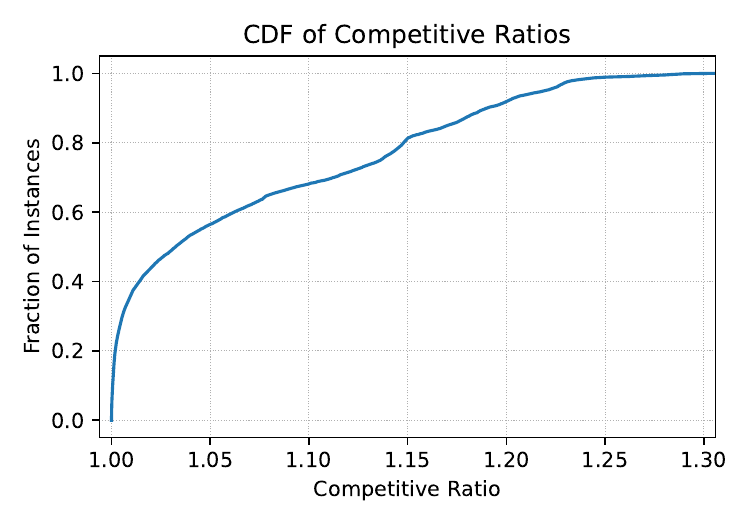}
  \caption{The distribution of competitive ratios across all instances for $m=25$, shown as the fraction of instances with competitive ratio at most a given value. In total $33$ instances were evaluated.
  \label{fig:plot_distribution}}
  \end{minipage}
\hfill
    \begin{minipage}[t]{0.48\textwidth}
    \includegraphics[width=\textwidth]{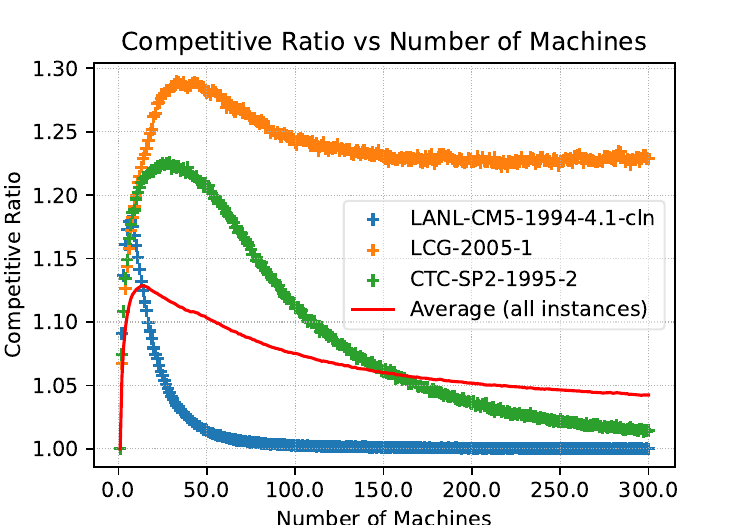}
  \caption{Average competitive ratio of \emph{Greedy-Balanced} over all instances as well as the ratio for three instances chosen to illustrate the variability across the dataset, as a function of the number of machines. 
  \label{fig:plot_machines}}
  
    \end{minipage} 
  \end{figure}

\FloatBarrier

\section{Discussion}\label{sec:discussion}

Our work advocates a particular way of incorporating fairness into optimization problems, namely by treating fairness notions as explicit constraints rather than as objectives to be optimized. 
This viewpoint leads to a different class of algorithmic questions: rather than asking whether a fair solution exists, we ask how much efficiency must be sacrificed to guarantee fairness, and whether this loss can be bounded by a constant. 
Our results show that, at least for interval scheduling, fairness can be enforced at a controlled and well-characterized cost, and that this cost persists even in structured settings. 
More broadly, we believe that this perspective provides a systematic way of integrating fairness considerations into scheduling, resource allocation, and other optimization problems.

Another insight emerging from our study is the importance of separating different sources of inefficiency. 
In the offline setting, we isolate the efficiency loss intrinsic to enforcing fairness itself, while in the online setting additional loss arises from uncertainty and irrevocable decisions. 
By comparing tight bounds across these settings, we distinguish which limitations stem from fairness and which are consequences of operating online. 
At the same time, our work leaves several natural directions open. 
We treat fairness as a hard constraint, whereas some applications may call for softer notions that allow controlled fairness violations in exchange for improved efficiency. 
We also focus on EF1, which coincides with EFX in our unweighted setting, while other fairness notions may exhibit qualitatively different algorithmic behavior.

Finally, our online results focus on the unweighted setting, since the weighted version of the general problem does not admit a constant competitive ratio. 
Nevertheless, weighted variants under additional structure, such as unit-length intervals with arbitrary weights, remain an interesting direction for future work. 
More generally, extending the fairness-as-constraint paradigm to other scheduling and online optimization problems may lead to a deeper understanding of the relationship between fairness, efficiency, and uncertainty. 
Moreover, our experimental evaluation suggests that the extremal instances driving the lower bounds may be rare in practice, indicating that fair algorithms could perform substantially better than their worst-case guarantees suggest.

\section*{Acknowledgments}

We thank Javier Cembrano for valuable discussions during the initial phase of this work. 
This collaboration was initiated at the Dagstuhl Seminar 24401, 
\emph{Fair Division: Algorithms, Solution Concepts, and Applications}. RV acknowledges support from DST
INSPIRE grant no. DST/INSPIRE/04/2020/000107, SERB/ANRF grant no. CRG/2022/002621,
and Mr. D.P. Gupta Chair Professorship.

\paragraph*{A note on the use of AI.}
An AI assistant (Claude, Anthropic; model \texttt{claude-fable-5}) helped with the preparation of this manuscript, drafting and revising parts of the exposition, producing the TikZ figures, and flagging typos and inconsistencies in earlier versions.
The content of the paper is the authors' own.
 
\clearpage

\bibliographystyle{plainnat}
\bibliography{ref}

\end{document}